\documentclass[10pt]{article}

\usepackage[paper=a4paper,left=20mm,right=20mm,top=25mm, bottom=25mm]{geometry}

\usepackage{hyperref}

\usepackage{authblk}
\usepackage{tikz}
\usepackage{xfrac}
\usepackage{tabularx, multirow}
\usepackage{graphicx, tikz}
\usetikzlibrary{arrows.meta}
\usetikzlibrary{arrows,calc}
\usetikzlibrary{shapes}
\usepackage{latexsym,amssymb,amsfonts,graphicx,makeidx,paralist}
\usepackage{theorem,rotating,ifthen,amsmath,subfigure,epsfig,nicefrac}
\usepackage{booktabs,framed}

\usepackage{hyperref}

\usepackage{xcolor}
\colorlet{shadecolor}{gray!12}

\newtheorem{theorem}{Theorem}[section]

\newtheorem{lemma}[theorem]{Lemma}
\newtheorem{observation}[theorem]{Observation}
\newtheorem{propostion}[theorem]{Proposition}
\newtheorem{definition}[theorem]{Definition}

\newenvironment{proof}{\noindent{\bf Proof~}}{\null\hfill $\Box$\par\medskip}

\newcommand{\bigo}{\text{$\mathcal O$}}

\newcommand{\un} {{\it un}}

\newcommand{\IR}{\mathbb{R}}

\newcommand{\forb} {\text{Forb}}
\newcommand{\free} {\text{Free}}
\newcommand{\DC} {\text{DC}}
\newcommand{\OC} {\text{OC}}
\newcommand{\DTP} {\text{DTP}}
\newcommand{\DCTP} {\text{DCTP}}
\newcommand{\OCTP} {\text{OCTP}}
\newcommand{\OTP} {\text{OTP}}
\newcommand{\DT} {\text{DT}}
\newcommand{\OT} {\text{OT}}
\newcommand{\FD} {\text{FD}}
\newcommand{\TD} {\text{TD}}
\newcommand{\C} {\text{C}}
\newcommand{\T} {\text{T}}
\newcommand{\TP} {\text{TP}}
\newcommand{\CTP} {\text{CTP}}
\newcommand{\SC} {\text{SC}}
\newcommand{\CSC} {\text{CSC}}
\newcommand{\WQT} {\text{WQT}}
\newcommand{\CWQT} {\text{CWQT}}
\newcommand{\DSC} {\text{DSC}}
\newcommand{\OSC} {\text{OSC}}
\newcommand{\DCSC} {\text{DCSC}}
\newcommand{\OCSC} {\text{OCSC}}
\newcommand{\TT} {\text{TT}}
\newcommand{\SPO} {\text{SPO}}

\newcommand{\DWQT} {\text{DWQT}}

\newcommand{\OWQT} {\text{OWQT}}
\newcommand{\DCWQT} {\text{DCWQT}}
\newcommand{\OCWQT} {\text{OCWQT}}

\newcommand{\co} {\text{co-}}

\newcommand{\ahead}{triangle 45}
\newcommand{\bnode}{node[draw, circle, fill = black, scale = 0.7, outer sep = 1pt]}

\begin{document}

\title{On Characterizations for Subclasses of Directed Co-Graphs\thanks{An extended abstract of this paper appeared in the Proceedings of the 13th International Conference on Combinatorial Optimization and Applications (COCOA 2019) \cite{GKR19h}.}}

\author{Frank Gurski}
\author{Dominique Komander}
\author{Carolin Rehs}

\affil{\small University of  D\"usseldorf,
Institute of Computer Science, Algorithmics for Hard Problems Group,\newline 
40225 D\"usseldorf, Germany}

\maketitle

%%%%%%%%%%%%%%%%%%%%%%%%%%%%%%%%%%%%%%%%%%%%%%%%%%%%%%%%%%%%%%%%%%%%%%%%%%%
%%%%%%%%%%%%%%%%%%%%%%%%%%%%%%%%%%%%%%%%%%%%%%%%%%%%%%%%%%%%%%%%%%%%%%%%%%%
%%%%%%%%%%%%%%%%%%%%%%%%%%%%%%%%%%%%%%%%%%%%%%%%%%%%%%%%%%%%%%%%%%%%%%%%%%%

\begin{abstract}
Undirected co-graphs are those graphs which can be generated from
the single vertex graph by disjoint union and join operations.
Co-graphs are exactly the $P_4$-free graphs (where $P_4$ denotes
the path on $4$ vertices). Co-graphs itself and
several subclasses  haven been intensively studied. Among these are
trivially perfect graphs, threshold graphs, weakly quasi threshold
graphs, and simple co-graphs.

Directed co-graphs are precisely those digraphs which can be defined  from
the single vertex graph by applying
the disjoint union, order composition,  and series composition. By omitting
the series composition we obtain the subclass of oriented co-graphs which
has been analyzed by Lawler in the 1970s and the restriction to linear expressions
was recently studied by Boeckner.
There are only a few versions of subclasses of directed co-graphs until now.
By transmitting the restrictions of undirected subclasses to the directed classes, we define the corresponding subclasses for directed co-graphs.
We consider directed and oriented versions of threshold graphs, simple co-graphs, co-simple
co-graphs, trivially perfect graphs, co-trivially perfect graphs, weakly quasi threshold graphs and
co-weakly quasi threshold graphs. For all these classes we
provide characterizations by finite sets of minimal forbidden induced subdigraphs.
Further we analyze relations between these graph classes.

%Different characterizations of undirected threshold graphs lead to different
%ideas for the definition of directed threshold graphs.
%In this paper we present the three ideas of Boeckner \cite{Boe18}, Cloteaux
%et al.\ \cite{CDMS14} and Ferres digraphs \cite{MP95a}.
%We will find different characterizations of these three graph classes and
%compare them to each other as well as to the class of directed co-graphs,
%which were defined by Bechet et al.\ in \cite{BGR97}.
%Furthermore, we will analyze which directed graphs parameters are limited
%or unlimited for the different directed threshold graphs.

\end{abstract}

%%%%%%%%%%%%%%%%%%%%%%%%%%%%%%%%%%%%%%%%%%%%%%%%%%%%%%%%%%%%%%%%%%%%%%%%%%%
%%%%%%%%%%%%%%%%%%%%%%%%%%%%%%%%%%%%%%%%%%%%%%%%%%%%%%%%%%%%%%%%%%%%%%%%%%%
%%%%%%%%%%%%%%%%%%%%%%%%%%%%%%%%%%%%%%%%%%%%%%%%%%%%%%%%%%%%%%%%%%%%%%%%%%%

\section{Introduction}

%\cite{BLS99}

During the last years classes of directed graphs
have received a lot of attention \cite{BG18}, since they are useful in multiple applications of directed networks.
Meanwhile the class of directed co-graphs is used in applications in the field of genetics, see \cite{NEMWH18}.
But the field of directed co-graphs is far from been as well studied as the undirected version, even though it has a similar useful structure. There are multiple subclasses of undirected co-graphs which were characterized successfully by different definitions. Meanwhile there are also corresponding subclasses of directed co-graphs  as e.g. the class of oriented co-graphs,
which has been analyzed by Lawler \cite{Law76} and Boeckner \cite{Boe18}. But there are many more interesting subclasses of directed co-graphs, that were mostly not characterized until now. Thus we consider directed versions of threshold graphs, simple co-graphs, trivially perfect graphs and weakly quasi threshold graphs. Furthermore, we take a look at the oriented versions of these classes and the related complement classes. All of these classes are hereditary, just like directed co-graphs, such that they can be characterized by a set of forbidden induced subdigraphs. We will even show a finite number of forbidden induced subdigraphs for the further introduced classes. This is for example very useful in the case of finding a polynomial recognition algorithm for these classes.

Undirected co-graphs, i.e.\ complement reducible graphs, appeared independently
by several authors, see \cite{Ler71,Sum74} for example, while directed co-graphs were
introduced 30 years later by Bechet et al.\ \cite{BGR97}. Due to their recursive structure
there are problems, that are hard in general, which can be solved efficiently on (directed) co-graphs,
see \cite{BM93,CLS81,CPS84,LOP95,GHKRRW20a} and \cite{BM14,Gur17a,GKR19f,GKR19d,GR18c,GHKRRW20,GKL20}.
That makes both graph classes particularly interesting. This paper can be seen a basis for further research on possible algorithms that are efficient for the structures of the presented graph classes.

This paper is organized as follows.
After introducing some basic definitions we introduce undirected co-graphs in Section  \ref{undco} and subclasses and recapitulate their relations and their characterizations by sets of forbidden subgraphs.
In Section \ref{dirco} we introduce directed and oriented co-graphs and summarize their properties. Subsequently, we focus on subclasses of directed co-graphs. We show definitions of series-parallel partial order digraphs, directed trivially perfect graphs, directed weakly quasi threshold graphs, directed simple co-graphs, directed threshold graphs and the corresponding complementary and oriented versions of these classes. Some of the subclasses already exist, others are motivated by the related subclasses of undirected co-graphs given in Table \ref{t-uc}. All of these subclasses have in common that they can be constructed recursively by several operations. Analogously to the undirected classes, we show how these multiple subclasses can be characterized by finite sets of minimal forbidden induced subdigraphs. We continue with an analysis of the relations of the several classes. Moreover, we analyze how they are related to the corresponding undirected classes.
Finally in Section \ref{sec-con}, we give conclusions including further research
directions.

%%%%%%%%%%%%%%%%%%%%%%%%%%%%%%%%%%%%%%%%%%%%%%%%%%%%%%%%%%%%%%%%%%%%%%%%%%%
\section{Preliminaries}\label{intro}
%%%%%%%%%%%%%%%%%%%%%%%%%%%%%%%%%%%%%%%%%%%%%%%%%%%%%%%%%%%%%%%%%%%%%%%%%%%

\subsection{Notations for Undirected Graphs}

We work with finite undirected {\em graphs} $G=(V,E)$,
where $V$ is a finite set of {\em vertices}
and $E \subseteq \{ \{u,v\} \mid u,v \in
V,~u \not= v\}$ is a finite set of {\em edges}.
%For a vertex $v\in V$ we denote by $N_G(v)$
%the set of all vertices which are adjacent to $v$ in $G$,
%i.e. $N_G(v)=\{w\in V~|~\{v,w\}\in E\}$.
%Set $N_G(v)$ is called the set of all {\em neighbors} of $v$
%in $G$ or {\em neighborhood} of $v$ in $G$.
%The {\em degree} of a vertex $v\in V$,  denoted by $\deg_G(x)$,
%is the number of neighbors of vertex $v$ in $G$, i.e. $\deg_G(v)=|N_G(v)|$.
A graph $G'=(V',E')$ is a {\em subgraph} of graph $G=(V,E)$ if $V'\subseteq V$
and $E'\subseteq E$.  If every edge of $E$ with both end vertices in $V'$  is in
$E'$, we say that $G'$ is an {\em induced subgraph} of digraph $G$ and
we write $G'=G[V']$.

For some  graph $G=(V,E)$ its complement digraph is defined by
$$\text{co-}G=(V,\{\{u,v\}~|~\{u,v\}\not\in E, u,v\in V, u\neq v\}).$$

For some graph class $X$ we define by $\co X=\{\co G ~|~ G\in X\}$.

For some graph $G$ some integer $d$ let $d G$ be the disjoint union
of $k$ copies of $G$.

\paragraph{Special Undirected Graphs}
As usual we denote by $$K_n=(\{v_1,\ldots,v_n\},\{\{v_i,v_j\}~|~1\leq i<j\leq n\}),$$
$n \ge 1$ a complete graph on $n$ vertices and by $I_n$ an edgeless graph on
$n$ vertices, i.e. the complement graph of a complete graph on $n$ vertices.
By $$P_n=(\{v_1,\ldots,v_n\},\{\{v_1,v_2\},\ldots, \{v_{n-1},v_n\}\})$$
we denote a path on $n$ vertices. See Table \ref{gr} for examples.

\subsection{Notations for Directed Graphs}

A {\em directed graph} or {\em digraph} is a pair  $G=(V,E)$, where $V$ is
a finite set of {\em vertices} and  $E\subseteq \{(u,v) \mid u,v \in
V,~u \not= v\}$ is a finite set of ordered pairs of distinct\footnote{Thus, we do not consider
directed graphs with loops or multiple edges.} vertices called {\em arcs}.
%For a vertex $v\in V$, the sets $N_D^+(v)=\{u\in V~|~ (v,u)\in A\}$ and
%$N_D^-(v)=\{u\in V~|~ (u,v)\in A\}$ are called the {\em set of all out-neighbours}
%and the {\em set of all in-neighbours} of $v$.
%The  {\em outdegree} of $v$, $\odeg_D(v)$ for short, is the number
%of out-neighbours of $v$ and the  {\em indegree} of $v$, $\ideg_D(v)$ for short,
%is the number of in-neighbours of $v$ in $D$.
A vertex  $v\in V$ is  {\em out-dominating (in-dominated)} if
it is adjacent to every other vertex in $V$ and is a source (a sink, respectively).
A digraph $G'=(V',E')$ is a {\em subdigraph} of digraph $G=(V,E)$ if $V'\subseteq V$
and $E'\subseteq E$.  If every arc of $E$ with both end vertices in $V'$  is in
$E'$, we say that $G'$ is an {\em induced subdigraph} of digraph $G$ and we
write $G'=G[V']$.

For a directed graph $G=(V,E)$ its complement digraph is defined by
$$\text{co-}G=(V,\{(u,v)~|~(u,v)\not\in E, u,v\in V, u\neq v\})$$
and its {\em converse digraph} is defined by
$$G^c=(V,\{(u,v)~|~(v,u)\in E, u,v\in V, u\neq v\}).$$
For a digraph class $X$ we define by $\co X=\{\co G ~|~ G\in X\}$.
For a digraph $G$ and an integer $d$ let $d G$ be the disjoint union
of $k$ copies of $G$.

\paragraph{Orientations}
There are several ways to define a digraph $D=(V,A)$ from an undirected
graph $G=(V,E)$, see \cite{BG09}.
If we replace every edge $\{u,v\}$ of $G$ by
\begin{itemize}
\item
one of the arcs $(u,v)$ and $(v,u)$, we denote $D$ as an {\em orientation} of $G$.
Every digraph $D$  that can be obtained by an orientation of an undirected
graph $G$ is called an {\em oriented graph}.

\item
one or both of the arcs $(u,v)$ and $(v,u)$, we denote $D$ as a {\em biorientation} of $G$.
A digraph $D$ that we get by a biorientation of an undirected
graph $G$ is called a {\em bioriented graph}.

\item
both arcs $(u,v)$ and $(v,u)$, we denote $D$ as a {\em complete biorientation} of $G$.
Since in this case $D$ is well defined by $G$ we also denote
it by $\overleftrightarrow{G}$.
A digraph $D$  that we obtain by a complete biorientation of an undirected
graph $G$ is called a {\em complete bioriented graph}.
\end{itemize}

For a given digraph $D=(V,A)$ we define
its underlying undirected graph by ignoring the directions of the edges and deleting multiple edges, i.e.
$$\un(D)=(V,\{\{u,v\}~|~(u,v)\in A, u,v\in V\}).$$
and for some class of digraphs $X$, let
$$\un(X)=\{\un(D)~|~ D\in X\}.$$

\paragraph{Special Directed Graphs}
As usual we denote by $$\overleftrightarrow{K_n}=(\{v_1,\ldots,v_n\},\{ (v_i,v_j)~|~1\leq i\neq j\leq n\}),$$
$n \ge 1$ a bidirectional complete digraph
on $n$ vertices and by $\overleftrightarrow{I_n}$ an edgeless graph with
$n$ vertices, i.e. the complement graph of a complete directed graph on $n$ vertices.
By $$\overleftrightarrow{K_{n,m}}=(\{v_1,\ldots,v_n,w_1,\ldots,w_m\},\{ (v_i,w_j),(w_j,v_i)~|~1\leq i\leq n,1\leq  j\leq m\}),$$
$n,m \ge 1$ a bidirectional complete bipartite digraph
with $n+m$ vertices.
A \emph{directed clique} is a bidirectional complete digraph, such that $K=(V,E)$ with 
$E=\{(u,v) \mid \forall u,v \in V,  u \neq v \}$ holds.

\paragraph{Special Oriented Graphs}
By $$\overrightarrow{P_n}=(\{v_1,\ldots,v_n\},\{ (v_1,v_2),\ldots, (v_{n-1},v_n)\})$$
we denote the oriented path on $n$ vertices. By
$$\overrightarrow{C_n}=(\{v_1,\ldots,v_n\},\{ (v_1,v_2),\ldots, (v_{n-1},v_n),(v_n,v_1)\})$$
we denote the oriented cycle on $n$ vertices. By $T_n$ we denote a transitive tournament
on $n$ vertices.\footnote{Please note that transitive tournaments on $n$ vertices
are unique up to isomorphism,
see \cite[Chapter 9]{Gou12}, \cite{GRR18} and  Lemma \ref{le-tt}.}
By $$\overrightarrow{K_{n,m}}=(\{v_1,\ldots,v_n,w_1,\ldots,w_m\},\{ (v_i,w_j)~|~1\leq i\leq n,1\leq  j\leq m\}),$$
$n,m \ge 1$ we denote an oriented complete bipartite digraph
on $n+m$ vertices.

%%%%%%%%%%%%%%%%%%%%%%%%%%%%%%%%%%%%%%%%%%%%%%%%%%%%%%%%%%%%%%%%%%%%%%%%%%%
\subsection{Induced Subgraph Characterizations for Hereditary Classes}\label{forb}
%%%%%%%%%%%%%%%%%%%%%%%%%%%%%%%%%%%%%%%%%%%%%%%%%%%%%%%%%%%%%%%%%%%%%%%%%%%

The following notations and results are given in \cite[Chapter 2]{KL15} for undirected
graphs. These results also hold  for directed graphs.

Classes of (di)graphs which are closed  under taking induced sub(di)graphs
are called  {\em hereditary}. For some (di)graph class $F$ we define
$\free(F)$ as the set of all (di)graphs $G$ such that no induced sub(di)graph of $G$ is isomorphic
to a member of  $F$.

\begin{theorem}[\cite{KL15}]
A class of (di)graphs $X$ is hereditary if and only if there is a set
$F$, such that $\free(F)=X$.
\end{theorem}

A (di)graph $G$ is a {\em minimal forbidden induced sub(di)graph}
for some hereditary class $X$ if $G$ does not belong to $X$ and every
proper induced sub(di)graph of $G$ belongs to $X$.
For some hereditary (di)graph class $X$ we define
$\forb(X)$ as the set of all minimal forbidden induced sub(di)graphs
for $X$.

\begin{theorem}[\cite{KL15}]\label{th-free-forb}
For every hereditary  class of (di)graphs $X$ it holds that
$X=\free(\forb(X))$.
Set $\forb(X)$ is unique and of minimal size.
\end{theorem}

\begin{theorem}[\cite{KL15}]\label{th-free-sub}
$\free(F_1)\subseteq \free(F_2)$ if and only if for every (di)graph $G\in F_2$
there is a (di)graph $H\in F_1$ such that $H$ is an induced sub(di)graph of $G$.
\end{theorem}

\begin{lemma}[\cite{KL15}]\label{th-free}
Let  $X=\free(F_1)$ and $Y=\free(F_2)$ be hereditary  classes of (di)graphs.
Then $X\cap Y=\free(F_1\cup F_2)$ and $\co X=\free(\co F_1)$.
\end{lemma}

\begin{observation} \label{forb-u-d1}
Let $G$ be a digraph such that $G\in \free(X)$ for a hereditary class of digraphs $\free(X)$ and
there is some digraph $X^*\in X$ such that  all biorientations of $\un(X^*)$
are in $\free(X)$, then $\un(G)\in\free(\un(X^*))$.
\end{observation}

\begin{observation} \label{forb-u-d2}
Let $G$ be a digraph such that
$\un(G)\in \free(X)$ for some hereditary class of graphs $\free(X)$,
then  for all $X^*\in X$ and all biorientations $b(X^*)$ of $X^*$ it holds that
$G\in\free(b(X^*))$.
\end{observation}

%%%%%%%%%%%%%%%%%%%%%%%%%%%%%%%%%%%%%%%%%%%%%%%%%%%%%%%%%%%%%%%%%%%%%%%%%%%
\section{Undirected Co-graphs and Subclasses}\label{undco}
%%%%%%%%%%%%%%%%%%%%%%%%%%%%%%%%%%%%%%%%%%%%%%%%%%%%%%%%%%%%%%%%%%%%%%%%%%%

In order to define co-graphs and subclasses we will use the following
operations.
Let $G_1=(V_1,E_1)$ and $G_2=(V_2,E_2)$ be two vertex-disjoint graphs.
\begin{itemize}
\item
The disjoint union of $G_1$ and $G_2$, denoted by $G_1\oplus G_2$,
is the graph with vertex set $V_1\cup V_2$ and edge set $E_1\cup E_2$.
\item
The join of $G_1$ and $G_2$, denoted by $G_1\otimes G_2$, is the graph
with vertex set $V_1\cup V_2$ and
edge set $E_1\cup E_2 \cup \{\{v_1,v_2\}~|~ v_1\in V_1, v_2\in V_2\}$.
\end{itemize}

We also will recall forbidden induced subgraph characterizations for
co-graphs and frequently analyzed subclasses. Therefore the graphs
in Table \ref{gr} are very useful.

\begin{table}[ht]
\caption{Special graphs}
\label{gr}
\begin{center}
\begin{tabular}{ccccccc}
\\
\epsfig{figure=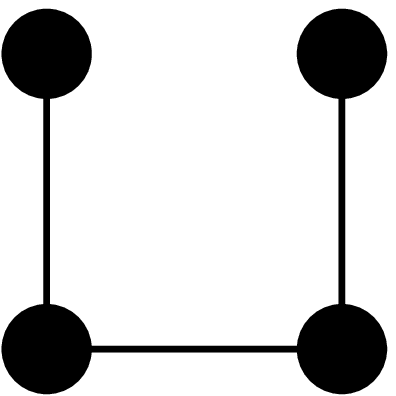,width=1.2cm} && \epsfig{figure=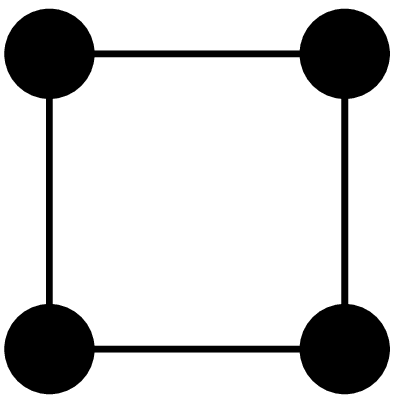,width=1.2cm} && \epsfig{figure=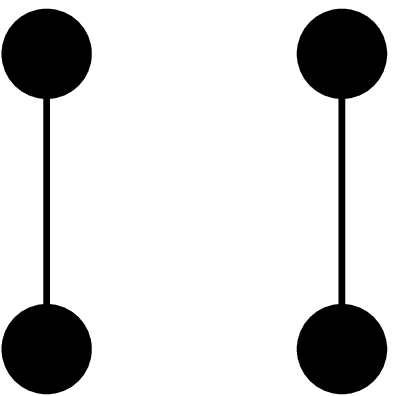,width=1.2cm}       && \epsfig{figure=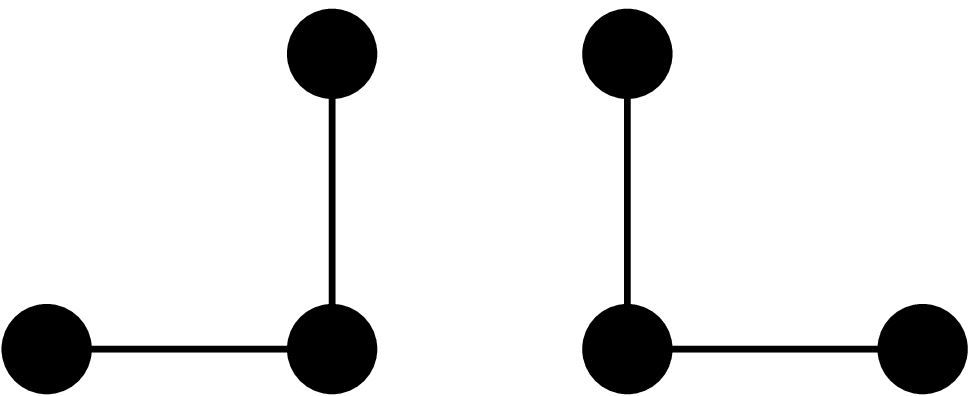,width=3.0cm}\\
$P_4$&& $C_4$ && $2K_2$ && $2P_3$\\
\end{tabular}
\end{center}
\end{table}

%%%%%%%%%%%%%%%%%%%%%%%%%%%%%%%%%%%%%%%%%%%%%%%%%%%%%%%%%%%
\subsection{Co-Graphs}\label{sec-cg}
%%%%%%%%%%%%%%%%%%%%%%%%%%%%%%%%%%%%%%%%%%%%%%%%%%%%%%%%%%%

\begin{definition}[Co-Graphs \cite{CLS81}]
The class of {\em co-graphs  (short for complement-reducible graphs)} is defined
recursively as follows.
\begin{enumerate}[(i)]
\item Every graph on a single vertex $(\{v\},\emptyset)$,
denoted by $\bullet$, is a co-graph.

\item If $G_1$ and $G_2$ are co-graphs, then
\begin{inparaenum}[(a)]
\item
$G_1\oplus G_2$ and
\item
$G_1\otimes  G_2$  are co-graphs.
\end{inparaenum}
\end{enumerate}
The class of  co-graphs is denoted by $\C$.
\end{definition}

Using the recursive structure of co-graphs many problems can be solved
in linear time, e.g. see \cite{CLS81}.
Further the recursive structure allows to compute the path-width
and tree-width of co-graphs in linear time \cite{BM93}.

%%%%%%%%%%%%%%%%%%%%%%%%%%%%%%%%%%%%%%%%%%%%%%%%%%%%%%%%%%%
\subsection{Subclasses of Co-Graphs}
%%%%%%%%%%%%%%%%%%%%%%%%%%%%%%%%%%%%%%%%%%%%%%%%%%%%%%%%%%%

In Table \ref{t-uc} we summarize co-graphs and their well-known subclasses.
The given forbidden sets are known from the existing literature \cite{CLS81,Gol78,CH77,NP11,HMP11}.

\begin{table}[ht]
\caption{Overview on subclasses of co-graphs. By $G_1$ and $G_2$ we denote graphs
of the class $X$, by $I$ we denote an edgeless graph and by $K$ we denote a complete graph.
Classes and complement classes are listed
between two horizontal lines. That is, only co-graphs and threshold graphs  are
closed under taking edge complementation.}
\label{t-uc}
{\small
\begin{center}
\begin{tabular}{|l|c|ccc|l|}
\hline
class  $X$              &  notation   & \multicolumn{3}{c|}{operations}    & $\forb(X)$  \\
%                       &            &             &    &                 &   \\
\hline
\hline
co-graphs              & $\C$  & $\bullet$  & $G_1\oplus G_2$     & $G_1\otimes G_2$  & $P_4$   \\
\hline
\hline
quasi threshold/trivially perfect  graphs    & $\TP$ & $\bullet$  & $G_1\oplus G_2$     & $G_1 \otimes \bullet$ & $P_4$, $C_4$\\
\hline
co-quasi threshold/co-trivially perfect graphs    & $\CTP$ & $\bullet$  &   $G_1 \oplus \bullet$                   &    $G_1\otimes G_2$             & $P_4$,  $2K_2$   \\
\hline
\hline
threshold graphs       & $\T$  & $\bullet$  & $G_1\oplus \bullet$ & $G_1 \otimes \bullet$ & $P_4$, $C_4$, $2K_2$  \\
\hline
\hline
simple co-graphs       & $\SC$ & $\bullet$  &  $G_1\oplus I$      & $G_1 \otimes I$ & $P_4$, co-$2P_3$, $2K_2$ \\
\hline
co-simple co-graphs                     &  $\CSC$ & $\bullet$  &  $G_1\oplus K$      & $G_1 \otimes K$ & $P_4$, $2P_3$, $C_4$ \\
\hline \hline
weakly quasi threshold  graphs &  $\WQT$ & $I$        &   $G_1\oplus G_2$   &  $G_1 \otimes I$             & $P_4$,  co-$2P_3$    \\
\hline
co-weakly quasi threshold graphs & $\CWQT$ & $K$        &   $G_1\oplus K$   &  $G_1\otimes G_2$          & $P_4$,  $2P_3$    \\
\hline
\hline
edgeless graphs &  & $\bullet$        &   $G_1\oplus \bullet$   &              & $P_2$    \\
\hline
complete graphs &  & $\bullet$        &   &$G_1\otimes   \bullet$               & co-$P_2$   \\
\hline
\hline
disjoint union of two cliques &    &  & $K\oplus K$  && $I_3$ ,$P_3$\\
\hline
complete bipartite graphs &        &    &  &$I \otimes I$    & $K_3$,  co-$P_3$\\
\hline
\hline
disjoint union of cliques &    & $K$ & $G_1\oplus K$  &&  $P_3$\\
\hline
join of stable sets &  &$I$   &  &$G_1 \otimes I$   &  co-$P_3$\\
\hline
\end{tabular}
\end{center}
}
\end{table}

%%%%%%%%%%%%%%%%%%%%%%%%%%%%%%%%%%%%%%%%%%%%%%%%%%%%%%%%%%%
\subsection{Relations} \label{rel-u}
%%%%%%%%%%%%%%%%%%%%%%%%%%%%%%%%%%%%%%%%%%%%%%%%%%%%%%%%%%%

In Figure \ref{fig:und} we compare the above graph classes to each other and show
the hierarchy of the subclasses of co-graphs.

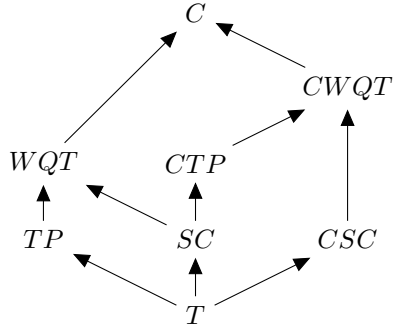
\begin{figure}[hbtp]
	\begin{center}
	\caption{Relations between the subclasses of co-graphs. If there is a path from $A$ to $B$, then it holds that $A\subset B$. The classes, that are not connected by a directed path are incomparable.}
	
	\medskip
		\begin{tikzpicture}
  % tree before
  % nodes

  % level 0
  \draw (3,0) node (T){$T$};
  %level 1
  \draw (1,1) node (TP){$TP$};
  \draw (3,1) node (SC){$SC$};
  \draw (5,1) node (CSC){$CSC$};

  % level 2
  \draw (1,2) node (WQT){$WQT$};
  \draw (3,2) node (CTP){$CTP$};

  % level 3
  \draw (5,3) node (CWQT){$CWQT$};

  % level 4
  \draw (3,4) node (C){$C$};

  % Kanten
  \draw [-\ahead] (T) edge (TP);
  \draw [-\ahead] (T) edge (SC);
  \draw [-\ahead] (T) edge (CSC);
  \draw [-\ahead] (TP) edge (WQT);
  \draw [-\ahead] (SC) edge (WQT);
  \draw [-\ahead] (SC) edge (CTP);
  \draw [-\ahead] (CSC) edge (CWQT);

  \draw [-\ahead] (WQT) edge (C);
  \draw [-\ahead] (CTP) edge (CWQT);
  \draw [-\ahead] (CWQT) edge (C);

\end{tikzpicture}
		
		\label{fig:und}
	\end{center}
\end{figure}

%%%%%%%%%%%%%%%%%%%%%%%%%%%%%%%%%%%%%%%%%%%%%%%%%%%%%%%%%%%%%%%%%%%%%%%%%%%
%%%%%%%%%%%%%%%%%%%%%%%%%%%%%%%%%%%%%%%%%%%%%%%%%%%%%%%%%%%%%%%%%%%%%%%%%%%
%%%%%%%%%%%%%%%%%%%%%%%%%%%%%%%%%%%%%%%%%%%%%%%%%%%%%%%%%%%%%%%%%%%%%%%%%%%
\section{Directed Co-Graphs and Subclasses}\label{dirco}
%%%%%%%%%%%%%%%%%%%%%%%%%%%%%%%%%%%%%%%%%%%%%%%%%%%%%%%%%%%%%%%%%%%%%%%%%%%
%%%%%%%%%%%%%%%%%%%%%%%%%%%%%%%%%%%%%%%%%%%%%%%%%%%%%%%%%%%%%%%%%%%%%%%%%%%
%%%%%%%%%%%%%%%%%%%%%%%%%%%%%%%%%%%%%%%%%%%%%%%%%%%%%%%%%%%%%%%%%%%%%%%%%%%

First we introduce operations in order to recall the definition of directed
co-graphs from \cite{BGR97} and introduce some interesting subclasses.
Let $G_1=(V_1,E_1)$ and $G_2=(V_2,E_2)$ be two vertex-disjoint directed graphs.\footnote{We use
the same symbols for the disjoint union and join between undirected and
directed graphs.  Although the meaning becomes clear from the context
we want to emphasize this fact.}
\begin{itemize}
\item
The {\em disjoint union} of $G_1$ and $G_2$,
denoted by $G_1 \oplus G_2$,
is the digraph with vertex set $V_1 \cup V_2$ and
arc set $E_1\cup E_2$.

\item
The {\em series composition} of $G_1$ and $G_2$,
denoted by $G_1\otimes G_2$,
is the digraph with vertex set $V_1 \cup V_2$ and
arc set $E_1\cup E_2\cup\{(u,v),(v,u)~|~u\in V_1, v\in V_2\}$.

\item
The {\em order composition} of $G_1$ and $G_2$,
denoted by $G_1\oslash  G_2$,
is the digraph with vertex set $V_1 \cup V_2$ and
arc set $E_1\cup E_2\cup\{(u,v)~|~u\in V_1, v\in V_2\}$.
\end{itemize}

Every graph structure which can be obtained by this operations, can
be constructed by a tree or even a sequence, as we could do for
undirected co-graphs and threshold graphs. These trees/sequences
can be used for algorithmic properties of those graphs.

%%%%%%%%%%%%%%%%%%%%%%%%%%%%%%%%%%%%%%%%%%%%%%%%%%%%%%%%%%%
\subsection{Directed Co-Graphs}\label{sec-dcg}
%%%%%%%%%%%%%%%%%%%%%%%%%%%%%%%%%%%%%%%%%%%%%%%%%%%%%%%%%%%

\begin{definition}[Directed Co-Graphs \cite{BGR97}]
The class of {\em directed co-graphs} is recursively defined as follows.
\begin{enumerate}[(i)]
\item Every digraph on a single vertex $(\{v\},\emptyset)$,
denoted by $\bullet$, is a directed co-graph.

\item If $G_1$ and $G_2$ are directed co-graphs, then
\begin{inparaenum}[(a)]
\item
$G_1\oplus G_2$,
\item
$G_1 \otimes G_2$, and
\item
$G_1\oslash G_2$   are directed co-graphs.
\end{inparaenum}
\end{enumerate}
The class of directed co-graphs is denoted by $\DC$.
\end{definition}

The recursive definition of directed and undirected co-graphs lead to the following
observation.

\begin{observation}\label{obs-co1}
For every directed co-graph $G$ the
underlying undirected graph $\un(G)$ is a co-graph.
\end{observation}

The reverse direction
only holds under certain conditions, see Theorem \ref{th-ch-dco}.

Obviously for every directed co-graph we can define a tree structure,
denoted as {\em binary di-co-tree}. The leaves of the di-co-tree represent the
vertices of the graph and the inner nodes of the di-co-tree  correspond
to the operations applied on the subexpressions defined by the subtrees.
For every directed co-graph one can construct a binary di-co-tree in linear time,
see \cite{CP06}.

In \cite{BM14} it is shown that the weak $k$-linkage problem can be solved
in polynomial time for directed co-graphs. By the recursive structure
there exist dynamic programming algorithms
to compute
the size of a largest edgeless subdigraph,
the size of a largest subdigraph which is a tournament,
the size of a largest semi complete subdigraph, and
the size of a largest complete subdigraph
for every directed co-graph  in linear time. Also
the hamiltonian path,  hamiltonian cycle, regular
subdigraph, and directed cut problem are polynomial
on directed co-graphs \cite{Gur17a}.
Calculs of directed co-graphs were also considered in connection
with pomset logic in \cite{Ret99}.
Further the directed path-width, directed tree-width, directed feedback vertex set number, 
cycle rank, DAG-depth and DAG-width
can be computed in linear time for  directed co-graphs \cite{GKR19f}.

\begin{lemma}[\cite{GR18b}]\label{le-com}
Let $G$ be some digraph, then the following properties hold.
\begin{enumerate}[1.]
\item
Digraph $G$ is a directed co-graph if and only if digraph $\co G$ is  a directed co-graph.

\item
Digraph $G$ is a directed co-graph if and only if digraph $G^c$ is  a directed co-graph.
\end{enumerate}
\end{lemma}

It further hold the following properties for directed co-graphs:

\begin{theorem}\label{th-ch-dco}
Let $G$ be a digraph. The following properties are equivalent:
\begin{enumerate}
	\item\label{td-1}   $G$ is a directed co-graph.
	\item\label{td-2}   $G\in \free(\{D_1, \dots, D_8\})$.
        \item\label{td-2aa} $G\in \free(\{D_1, \dots, D_6\})$  and $\un(G)\in\free(\{P_4\})$.
        \item\label{td-2a}  $G\in \free(\{D_1, \dots, D_6\})$  and $\un(G)$ is a co-graph.
	\item\label{td-3}   $G$ has directed NLC-width $1$.\footnote{For a definition of directed NLC-width, see \cite{GWY16}.}
	\item\label{td-4}   $G$ has directed clique-width at most $2$ and $G\in \free(\{D_2,D_3\})$.\footnote{For a definition of directed clique-width, see \cite{CO00}.}
\end{enumerate}
\end{theorem}

\begin{proof}
$(\ref{td-1})\Leftrightarrow(\ref{td-2})$ By \cite{CP06}.

$(\ref{td-2aa})\Leftrightarrow(\ref{td-2a})$ Since $\forb(\C)=\{P_4\}$.

$(\ref{td-1})\Leftrightarrow(\ref{td-3})$ By \cite{GWY16}

$(\ref{td-1})\Leftrightarrow(\ref{td-4})$ By \cite{GWY16}

$(\ref{td-2})\Rightarrow(\ref{td-2a})$ By Observation \ref{obs-co1}. 

$(\ref{td-2aa})\Rightarrow(\ref{td-2})$
by Observation \ref{forb-u-d2}.
\end{proof}

\begin{table}
\caption{The eight forbidden induced subdigraphs for directed co-graphs.}
\label{F-co}
\begin{center}
\begin{tabular}{cccccccc}
\epsfig{figure=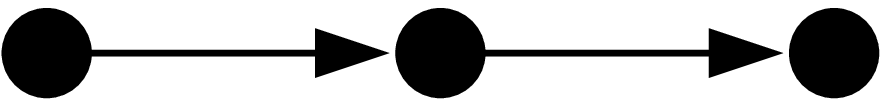,width=3.0cm} &&\epsfig{figure=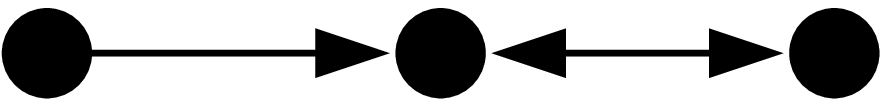,width=3.0cm}&&\epsfig{figure=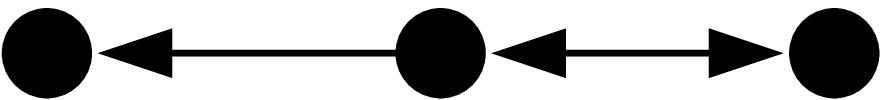,width=3.0cm}&&\epsfig{figure=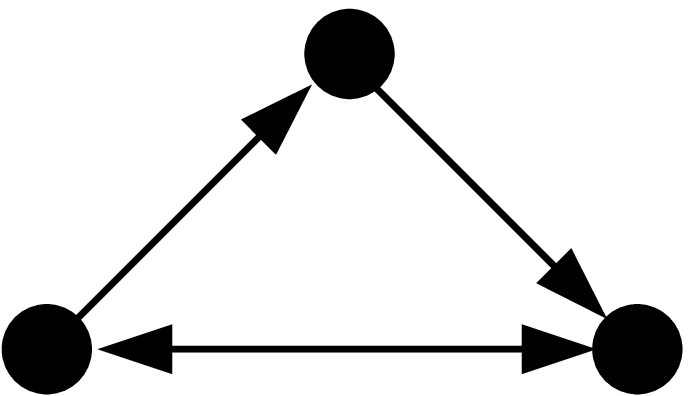,width=2.2cm} &\\
 $D_1$   &  &    $D_2$   &  &   $D_3$   &  &   $D_4$   &  \\
&&&&&&&\\
\epsfig{figure=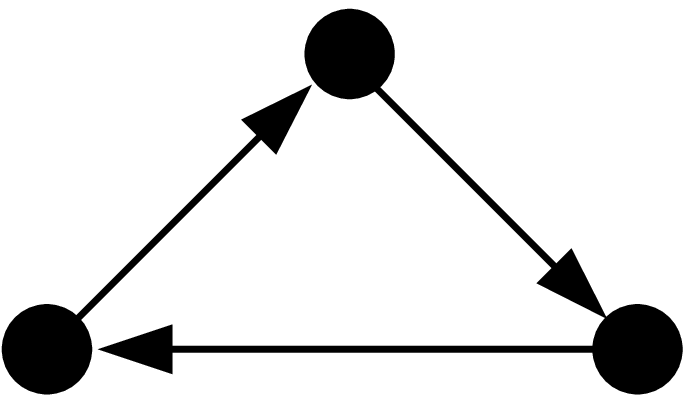,width=2.4cm} &&\epsfig{figure=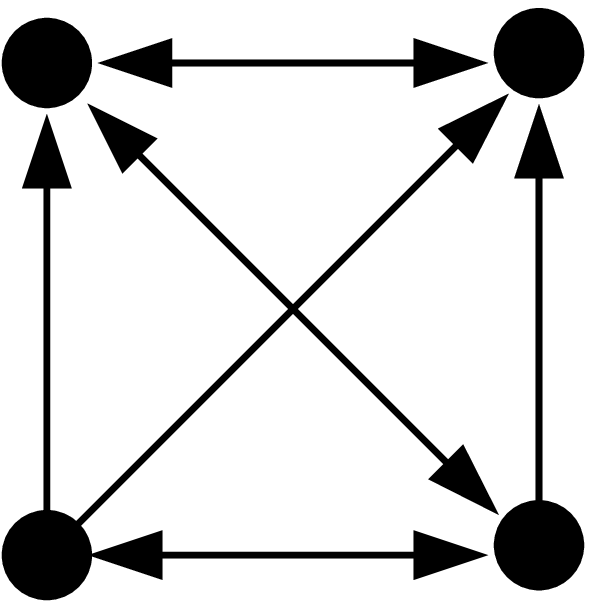,width=1.8cm}&&\epsfig{figure=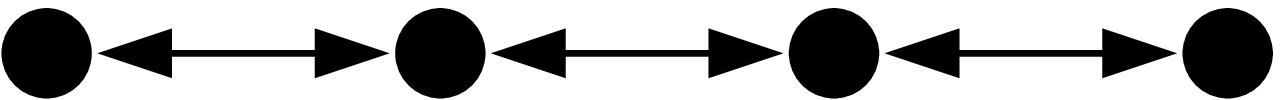,width=3.2cm}&&\epsfig{figure=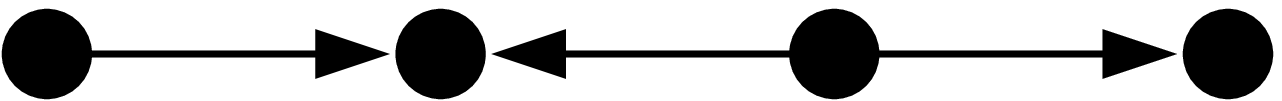,width=3.6cm}&\\
 $D_5$   &  &    $D_6$   &  &   $D_7$   &  &   $D_8$   &  \\
\end{tabular}
\end{center}
\end{table}

For subclasses of directed co-graphs, which will be defined in the following subsections,
some more forbidden subdigraphs are needed. Those are defined in Tables
 \ref{F-co2}, \ref{F-co2x}, and \ref{F-co2a}.

\begin{table}
\caption{Forbidden induced subdigraphs for subclasses of directed co-graphs.}
\label{F-co2}
\begin{center}
\begin{tabular}{ccccccccccccc}

\epsfig{figure=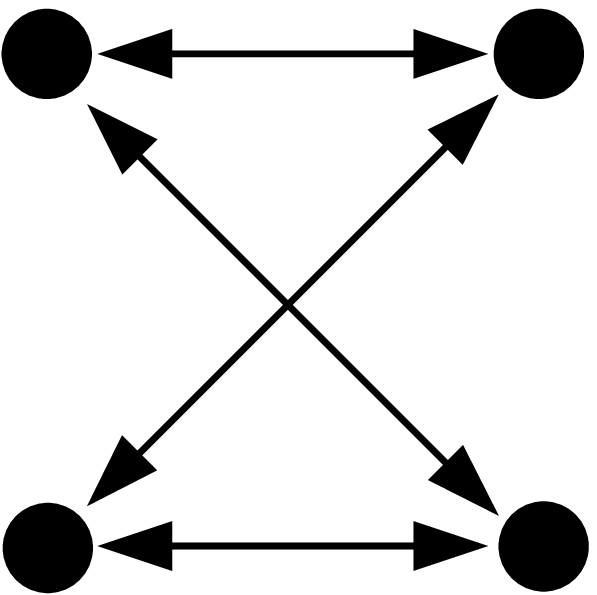,width=1.8cm} &~~&\epsfig{figure=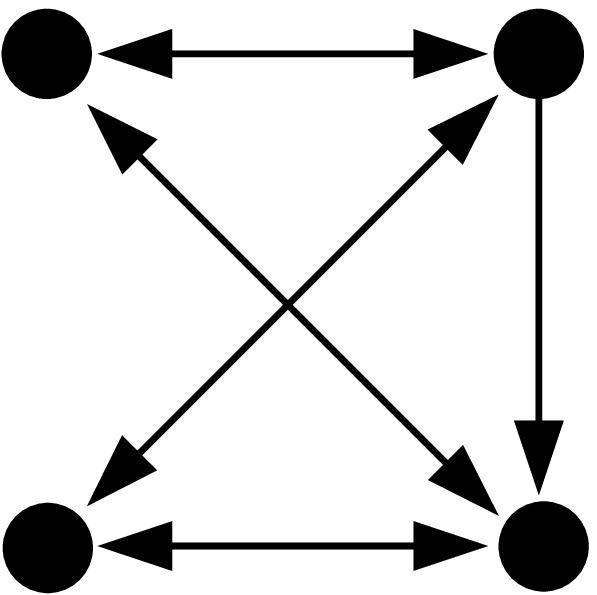,width=1.8cm} &~~& \epsfig{figure=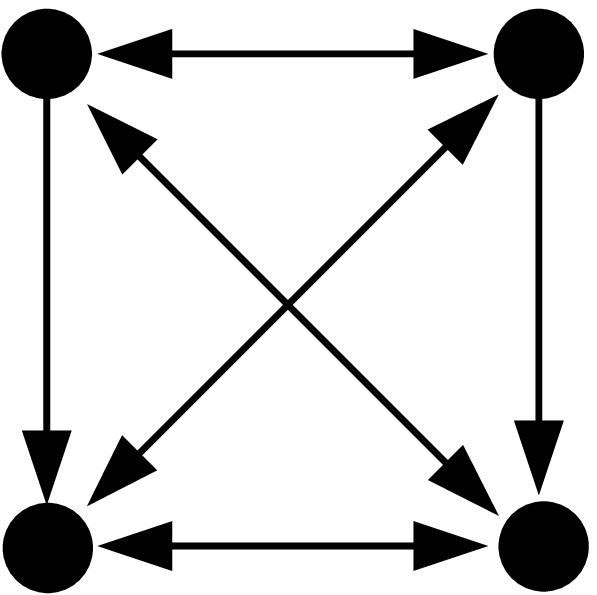,width=1.8cm} &~~& &  &~~&  &  \\
 $D_9$                             &&    $D_{10}$                        &&    $D_{11}$            &&    &    &&    &    \\
\end{tabular}
\end{center}
\end{table}

\begin{table}
\caption{Forbidden induced subdigraphs for subclasses of directed co-graphs.}
\label{F-co2x}
\begin{center}
\begin{tabular}{ccccccccccccc}

\epsfig{figure=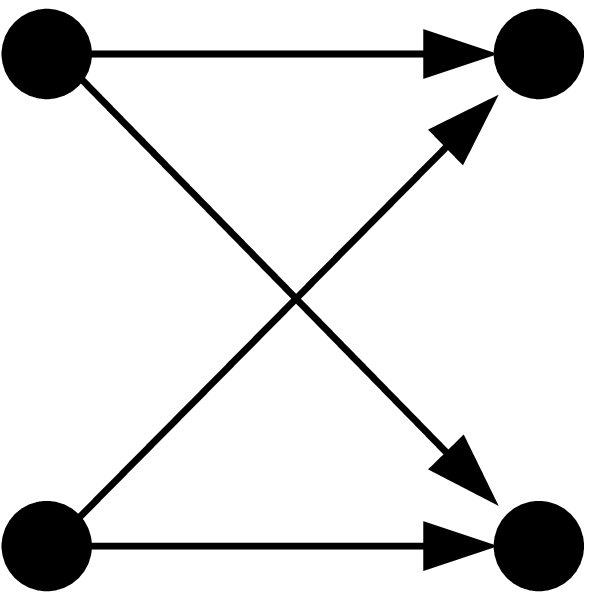,width=1.8cm} &~~&\epsfig{figure=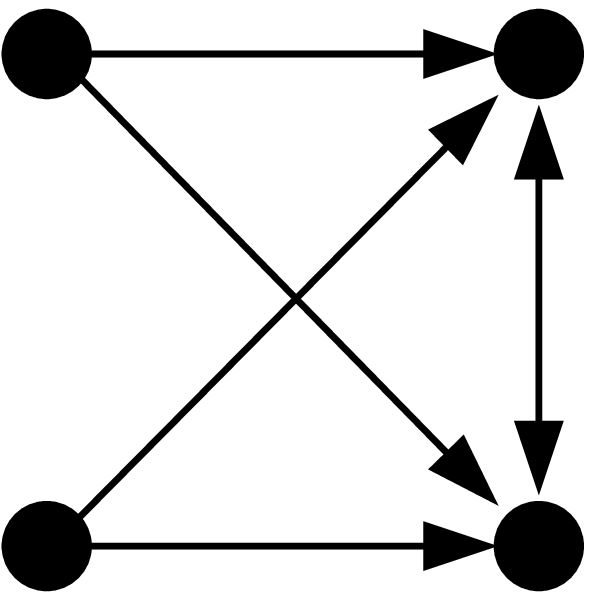,width=1.8cm} &~~& \epsfig{figure=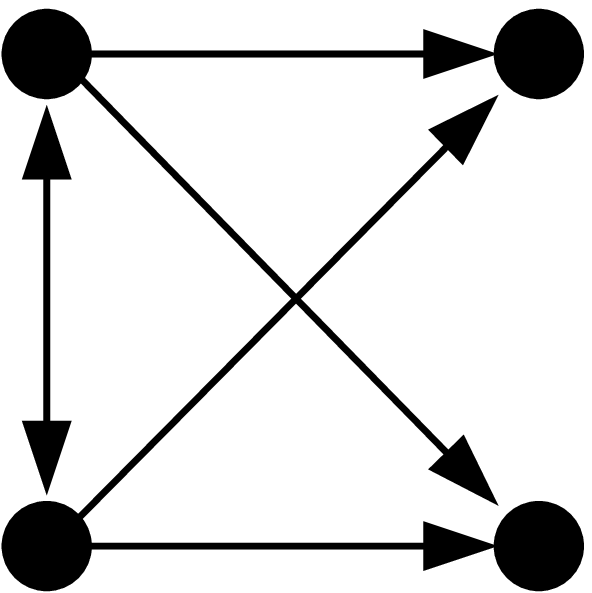,width=1.8cm} &~~&   \epsfig{figure=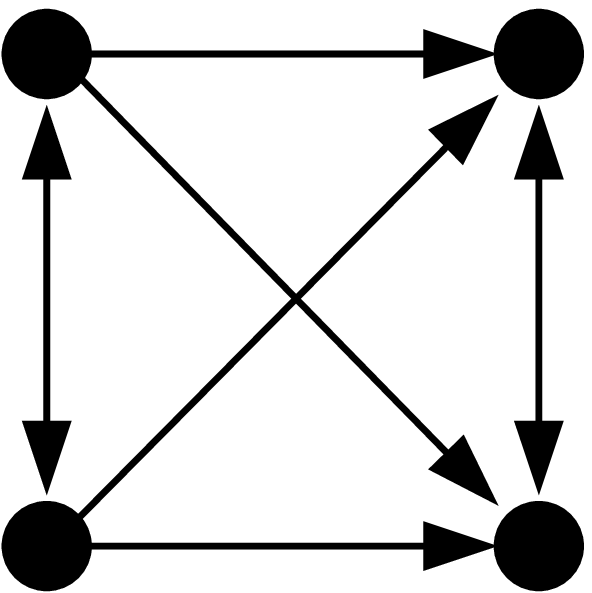,width=1.8cm}&~~&&  \\
 $D_{12}$                             &&    $D_{13}$                        &&    $D_{14}$            &&   $D_{15}$   &&   &\\
\end{tabular}
\end{center}
\end{table}

\begin{table}
\caption{Forbidden induced subdigraphs for subclasses of directed co-graphs.}
\begin{center}
\label{F-co2a}
\begin{tabular}{ccccccccccccc}
\epsfig{figure=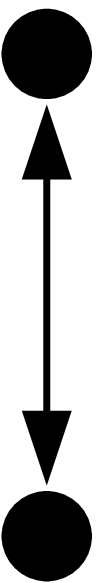,width=0.28cm} &~~&\epsfig{figure=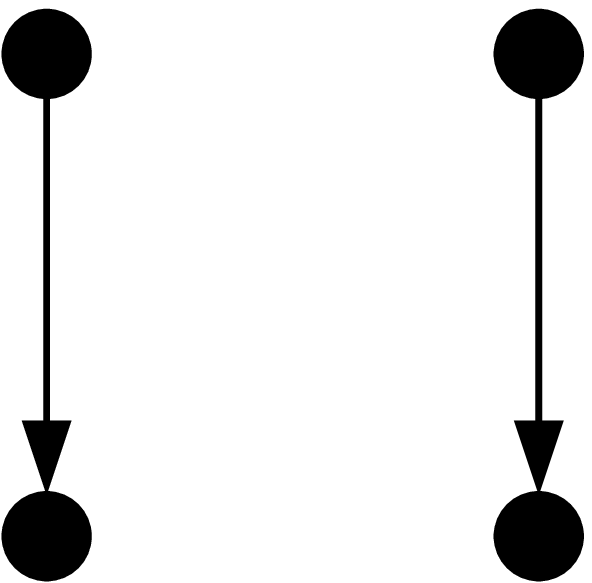,width=1.8cm} &~~&   \epsfig{figure=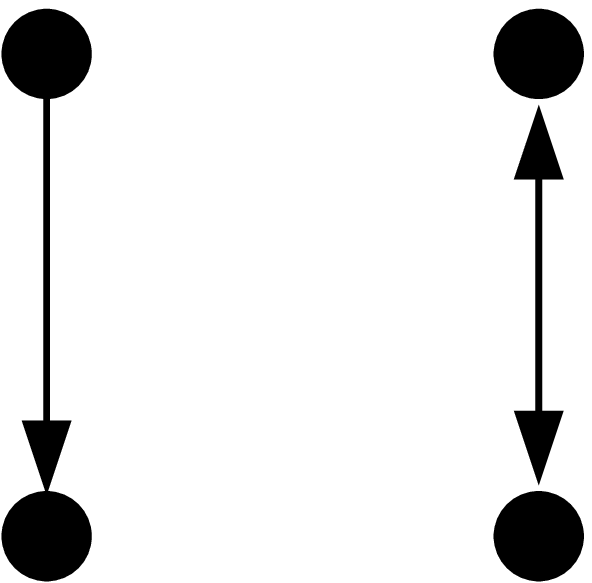,width=1.8cm}&~~& \epsfig{figure=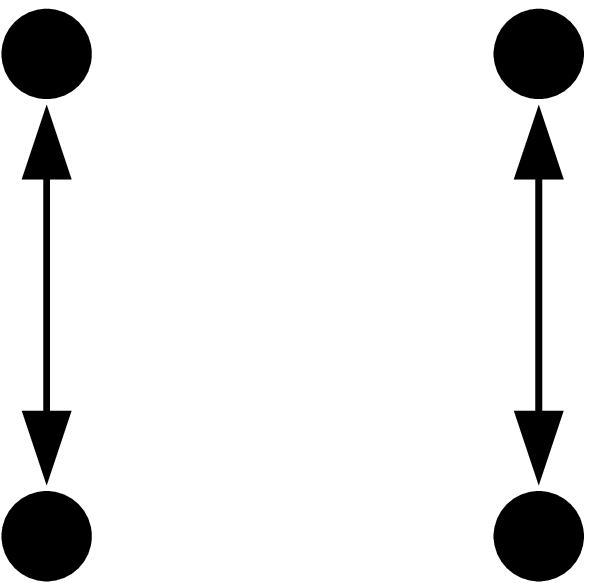,width=1.8cm} \\
$\overleftrightarrow{K_2}$                            && $2\overrightarrow{P_2}=\co D_{11}$           &&   $\co D_{10}$   && $\co D_9= \co D_9$   \\
\end{tabular}
\end{center}
\end{table}

\begin{observation}\label{obs18}
It holds:
\begin{enumerate}
\item
$\{D_1,\ldots, D_8\}=\co \{D_1,\ldots, D_8\}$.
\item
$\{D_{12},\ldots,D_{15}\}= \co \{D_{12},\ldots,D_{15}\}$.
\end{enumerate}
\end{observation}

For directed co-graphs Observation \ref{obs18} leads to the next result.

\begin{propostion}
$\DC=\co\DC$.
\end{propostion}

%%%%%%%%%%%%%%%%%%%%%%%%%%%%%%%%%%%%%%%%%%%%%%%%%%%%%%%%%%%
\subsection{Oriented Co-Graphs}\label{sec-ocog}
%%%%%%%%%%%%%%%%%%%%%%%%%%%%%%%%%%%%%%%%%%%%%%%%%%%%%%%%%%%

Beside directed co-graphs and their subclasses we also will
restrict the these classes to oriented graphs by
omitting the series operation.

\begin{definition}[Oriented Co-Graphs] \label{ocog}
The class of {\em oriented co-graphs} is recursively defined as follows.
\begin{enumerate}[(i)]
\item Every digraph on a single vertex $(\{v\},\emptyset)$,
denoted by $\bullet$, is an oriented co-graph.
\item If $G_1$ and $G_2$ are oriented co-graphs, then
\begin{inparaenum}[(a)]
\item
$G_1\oplus G_2$ and
\item
$G_1\oslash G_2$  are oriented co-graphs.
\end{inparaenum}
\end{enumerate}
The class of oriented co-graphs is denoted by $\OC$.
\end{definition}

The recursive definition of  oriented and undirected co-graphs lead to the following
observation.

\begin{observation}\label{obs-oco1}
For every oriented co-graph $G$ the
underlying undirected graph $\un(G)$ is a co-graph.
\end{observation}

The reverse direction only holds under certain conditions, see Theorem \ref{th-ch-oco}.
The class of  oriented co-graphs was already analyzed by Lawler
in \cite{Law76} and \cite[Section 5]{CLS81} using the notation
of {\em transitive series parallel (TSP)} digraphs.
A digraph $G=(V,A)$ is called {\em transitive} if for
every pair $(u,v)\in A$ and $(v,w)\in A$ of arcs
with $u\neq w$ the arc $(u,w)$ also belongs to $A$.
For oriented co-graphs  the oriented chromatic number
and also the graph isomorphism problem can be solved in
linear time \cite{GKR19d,GKL20}.

\begin{lemma}\label{le-co-t}
Let $G$ be a digraph such that $G\in\free(\{\overleftrightarrow{K_2},D_1,D_5\})$, then
$G$ is transitive.
\end{lemma}

\begin{proof}
Let  $(u,v),(v,w)\in A$ be two arcs of  $G=(V,A)$.
Since $G\in\free(\{\overleftrightarrow{K_2}\})$ we know that $(v,u),(w,v)\not\in A$.
Further since  $G\in\free(\{D_1,D_5\})$ we know that $u$ and $w$ are
connected either only by
$(u,w)\in A$ or by $(u,w)\in A$ and $(w,u)\in A$, which
implies that $G$ is transitive.
\end{proof}

The class $\OC$ can also be defined by forbidden subdigraphs.

\begin{theorem}\label{th-ch-oco}
Let $G$ be a digraph. The following properties are equivalent:
\begin{enumerate}
	\item\label{ch-oc1} $G$ is an oriented co-graph
	\item\label{ch-oc2} $G\in \free(\{D_1, D_5, D_8, \overleftrightarrow{K_2}\})$.
        \item\label{ch-oc3a} $G\in \free(\{D_1, D_5, \overleftrightarrow{K_2}\})$
		   and $\un(G)\in\free(\{P_4\})$.
        \item\label{ch-oc3} $G\in \free(\{D_1, D_5, \overleftrightarrow{K_2}\})$
		   and $\un(G)$ is a co-graph.

	\item\label{ch-oc4}  $G$ directed has NLC-width $1$ and $G\in \free(\{\overleftrightarrow{K_2}\})$.
	\item\label{ch-oc5}  $G$ has directed clique-width at most $2$ and $G\in \free(\{\overleftrightarrow{K_2}\})$.
        \item\label{ch-oc6} $G$ is  transitive and $G\in \free(\{\overleftrightarrow{K_2}, D_8\})$.
\end{enumerate}
\end{theorem}

\begin{proof}
$(\ref{ch-oc1})\Rightarrow(\ref{ch-oc2})$ If $G$ is an oriented co-graph, then
$G$ is a directed co-graph and by Theorem \ref {th-ch-dco} it holds that $G\in \free(\{D_1,\ldots, D_8\})$. Further $G\in \free(\{\overleftrightarrow{K_2}\})$
because of the missing series composition. This leads to $G\in \free(\{D_1, D_5, D_8, \overleftrightarrow{K_2}\})$.

$(\ref{ch-oc2})\Rightarrow(\ref{ch-oc1})$
If $G\in \free(\{D_1, D_5, D_8, \overleftrightarrow{K_2}\})$ then $G\in \free(\{D_1,\ldots, D_8\})$ and is a directed co-graph.
Since $G\in \free(\{\overleftrightarrow{K_2}\})$
there is no series operation in any
construction of $G$ which implies that $G$ is an oriented co-graph.

$(\ref{ch-oc3a})\Leftrightarrow(\ref{ch-oc3})$ Since $\forb(\C)=\{P_4\}$.

$(\ref{ch-oc2})\Rightarrow(\ref{ch-oc6})$ By Lemma \ref{le-co-t}
we know that $G$ is transitive.

$(\ref{ch-oc6})\Rightarrow(\ref{ch-oc2})$ If $G$ is transitive, then
$G\in\free(\{D_1, D_5\})$.

$(\ref{ch-oc1})\Leftrightarrow(\ref{ch-oc4})$ and  $(\ref{ch-oc1})\Leftrightarrow(\ref{ch-oc5})$ By Theorem \ref {th-ch-dco}.

$(\ref{ch-oc2})\Rightarrow(\ref{ch-oc3})$ By Observation \ref{obs-oco1}. 

$(\ref{ch-oc3a})\Rightarrow(\ref{ch-oc2})$
by Observation \ref{forb-u-d2}.
\end{proof}

\begin{observation}
Every oriented co-graph is a DAG.
\end{observation}

\begin{theorem}[\cite{CLS81}]
A graph $G$ is a co-graph if and only if there exists an orientation $G'$ of  $G$
such that $G'$ is an oriented co-graph.
\end{theorem}

%%%%%%%%%%%%%%%%%%%%%%%%%%%%%%%%%%%%%%%%%%%%%%%%%%%%%%%%%%%
\subsection{Series-parallel partial order digraphs}\label{sec-spd}
%%%%%%%%%%%%%%%%%%%%%%%%%%%%%%%%%%%%%%%%%%%%%%%%%%%%%%%%%%%

We recall the definitions of from \cite{BG18} which are based on \cite{VTL82}.
A series-parallel partial order is a partially ordered set $(X,\leq)$ that is
constructed by the series composition and
the parallel composition operation starting with a single element.
\begin{itemize}
\item
Let $(X_1 ,\leq)$ and $(X_2 , \leq)$ be two disjoint series-parallel partial
orders, then distinct elements $x,y \in X_1 \cup X_2$ of a series
composition\footnote{Note that the series composition in this case corresponds to the order composition in the definition of directed co-graphs.} have the same order they have in $X_1$ or $X_2$. Respectively, this holds if both of them are from the same set, and $x \leq  y$, if $x \in X_1$ and $y \in X_2$.
\item
Two elements $x, y \in X_1 \cup X_2$ of a parallel composition are comparable if and only if both of them are in $X_1$ or both in $X_2$, while they keep their corresponding order.
\end{itemize}

\begin{definition}[Series-parallel partial order digraphs]
A {\em series-parallel partial order digraph} $G=(V,E)$ is a digraph, where $(V,\leq)$
is a series-parallel partial order  and $(u,v)\in E$ if and only if
$u \neq v$ and $u \leq  v$.

The class of series-parallel partial order digraphs is denoted by $\SPO$.
\end{definition}

For a digraph $G=(V,E)$ an edge $(u,v)\in E$ is symmetric, if $(v,u)\in E$. Thus each bidirectional arc is symmetric. Further, an edge is asymmetric, if it is not symmetric, i.e. each edge with only one direction. We define the symmetric part of $G$ as $sym(G)$, which is the spanning subdigraph of $G$ that contains exactly the symmetric arcs of $G$. Analogously we define the asymmetric part of $G$ as $asym(G)$, which is the spanning subdigraph with only asymmetric edges.

Moreover, Bechet et al. showed in \cite{BGR97} the following property of directed co-graphs.

% Lemma? Oder lieber anders nennen?
\begin{lemma}[\cite{BGR97}]
  For every directed co-graph $G$ it holds that the asymmetric part of $G$ is a series-parallel partial order digraph and for the symmetric part the underlying undirected graph a co-graph.
\end{lemma}

The class of series-parallel partial ordered digraphs is equal to the class of oriented co-graphs, since they have exactly the same recursive structure. Thus this lemma is easy to prove with the following idea. Let $G$ be a directed co-graph and $T_G$ its corresponding di-co-tree.
\begin{itemize}
  \item
  \textbf{symmetric part:} Exchange each order composition with a directed union composition. Since there are no more oriented arcs left, this tree represents a co-graph.
  \item
  \textbf{asymmetric part:} Exchange each series composition with a directed union composition. Since there are no more bidirectional edges left, this tree represents an oriented co-graph, e.g. a series-parallel partial order digraph.
\end{itemize}

%%%%%%%%%%%%%%%%%%%%%%%%%%%%%%%%%%%%%%%%%%%%%%%%%%%%%%%%%%
\subsection{Directed trivially perfect graphs}\label{sec-dtpg}
%%%%%%%%%%%%%%%%%%%%%%%%%%%%%%%%%%%%%%%%%%%%%%%%%%%%%%%%%%%

\begin{definition}[Directed trivially perfect graphs] \label{dtpg}
The class of {\em directed trivially perfect graphs} is recursively defined as follows.
\begin{enumerate}[(i)]
\item Every digraph on a single vertex $(\{v\},\emptyset)$,
denoted by $\bullet$, is a  directed trivially perfect  graph.

\item If $G_1$ and $G_2$ are directed trivially perfect  graphs,
then $G_1\oplus G_2$  is a  directed trivially perfect  graph.

\item If $G$ is a  directed trivially perfect graph,
then
\begin{inparaenum}[(a)]
\item
$G\oslash \bullet$,
\item
$\bullet\oslash G$, and
\item
$G\otimes \bullet$  are  directed trivially perfect  graphs.
\end{inparaenum}

\end{enumerate}
The class of directed trivially perfect graphs is denoted by $\DTP$.
\end{definition}

The recursive definition of directed and undirected  trivially perfect graphs lead to the following
observation.

\begin{observation}\label{obs-tp1}
For every directed trivially perfect graph $G$ the
underlying undirected graph $\un(G)$ is a trivially perfect graph.
\end{observation}

The reverse direction
only holds under certain conditions, see Theorem \ref{ch-dtp}.

\begin{lemma}[\cite{GRR18}]\label{le-tt}
For every digraph $G$ the following statements are equivalent.
\begin{enumerate}
\item \label{s11b} $G$ is a transitive tournament.
\item \label{s11c} $G$ is an acyclic tournament.
\item \label{s11d} $G\in\free(\{\overrightarrow{C_3}\})$  and $G$ is a tournament.
\item \label{s11h} $G$ can be constructed from the one-vertex graph $K_1$ by repeatedly adding
an  out-dominating vertex.
\item \label{s11i} $G$ can be constructed from the one-vertex graph $K_1$ by repeatedly adding
an in-dominated  vertex.
\end{enumerate}
\end{lemma}

The class $\DTP$ can also be defined by forbidden induced subdigraphs.
It holds that

\begin{theorem}\label{ch-dtp}
Let $G$ be a digraph. The following properties are equivalent:
\begin{enumerate}
	\item\label{ch-dtp1}  $G$ is a directed trivially perfect graph.
	\item\label{ch-dtp2}  $G\in\free(\{D_1, \dots, D_{15}\})$.
        \item\label{ch-dtp4}  $G\in\free(\{D_1,D_2,D_3,D_4,D_5,D_6,D_{10},D_{11},D_{13},D_{14},D_{15}\})$
		   and $\un(G)\in \free(\{C_4,P_4\})$.
        \item\label{ch-dtp3}  $G\in\free(\{D_1,D_2,D_3,D_4,D_5,D_6,D_{10},D_{11},D_{13},D_{14},D_{15}\})$
		   and $\un(G)$ is a trivially perfect graph.
\end{enumerate}
\end{theorem}

\begin{proof}
$(\ref{ch-dtp1})\Rightarrow(\ref{ch-dtp2})$ The given forbidden digraphs  $D_1, \dots, D_{15}$ are not directed
trivially perfect graphs
and the set of directed  trivially perfect
graphs is closed under taking induced subdigraphs.

$(\ref{ch-dtp2})\Rightarrow(\ref{ch-dtp1})$
Since $G\in\free(\{D_1,D_2,D_3,D_4,D_5,D_6,D_7,D_8\})$
digraph $G$ is a directed co-graph  by \cite{CP06}
and thus has a construction using disjoint union, series composition, and
order composition.

Since $G\in\free(\{D_9,D_{10},D_{11}\})$ we know that for every series combination between two graphs on
at least two vertices at least one of the graphs is bidirectional complete.
Such a subgraph can be inserted by a number of feasible operations
for directed trivially perfect graphs.

Since $G\in\free(\{D_{12},D_{13},D_{14},D_{15}\})$ we know that
for every order combination between two graphs on
at least two vertices at least one of the graphs is a tournament.
Since $G\in\free(\{D_{5}\})=\free(\{C_{3}\})$ by Lemma \ref{le-tt} we even
know that at least one of the graphs is a transitive tournament.
Such a graph can be defined by a sequence of outdominating or indominating vertices
(Lemma \ref{le-tt}) which are also feasible operations
for directed trivially perfect graphs.

$(\ref{ch-dtp4}) \Leftrightarrow(\ref{ch-dtp3})$  Since $\forb(\TP)=\{C_4,P_4\}$.

$(\ref{ch-dtp2}) \Rightarrow(\ref{ch-dtp3})$ By Observation \ref{obs-tp1}. 

$(\ref{td-2aa})\Rightarrow(\ref{td-2})$
By Observation \ref{forb-u-d2}.
\end{proof}

For directed trivially perfect graphs Observation \ref{obs18} leads to the next result.

\begin{propostion}
$\DTP\neq\co\DTP$
\end{propostion}

This motivates us to consider the class of edge complements of
directed trivially perfect graphs.

%%%%%%%%%%%%%%%%%%%%%%%%%%%%%%%%%%%%%%%%%%%%%%%%%%%%%%%%%%%
%%%%%%%%%%%%%%%%%%%%%%%%%%%%%%%%%%%%%%%%%%%%%%%%%%%%%%%%%
%%%%%%%%%%%%%%%%%%%%%%%%%%%%%%%%%%%%%%%%%%%%%%%%%%%%%%%%%%
%%%%%%%%%%%%%%%%%%%%%%%%%%%%%%%%%%%%%%%%%%%%%%%%%%%%%%%%%%

\begin{definition}[Directed co-trivially perfect graphs] \label{dctpg}
The class of {\em directed co-trivially perfect graphs} is recursively defined as follows.
\begin{enumerate}[(i)]
\item Every digraph on a single vertex $(\{v\},\emptyset)$,
denoted by $\bullet$, is a  directed co-trivially perfect  graph.

\item If $G_1$ and $G_2$ are directed co-trivially perfect  graphs,
then $G_1\otimes G_2$  is a  directed trivially perfect  graph.

\item If $G$ is a  directed co-trivially perfect graph,
then
\begin{inparaenum}[(a)]
\item
$G\oslash \bullet$,
\item
$\bullet\oslash G$, and
\item
$G\oplus \bullet$  are  directed co-trivially perfect  graphs.
\end{inparaenum}

\end{enumerate}
The class of directed co-trivially perfect graphs is denoted by $\DCTP$.
\end{definition}

Theorem \ref{ch-dtp} and Lemma \ref{th-free} lead to the
following characterization for directed co-trivially perfect graphs.

\begin{theorem}
Let $G$ be a digraph. The following properties are equivalent:
\begin{enumerate}
	\item $G$ is a directed co-trivially perfect graph.
	\item $G\in \free(\{D_1, \dots, D_{8}, \co D_{11},\co D_{10},\co D_9,D_{12},\ldots,D_{15}\})$.
    \item $G\in \free(\{D_1,\ldots,D_6,D_{12},\ldots,D_{15}\})$ and  $\un(G)\in \free(\{P_4,2K_2\})$.
    \item $G\in \free(\{D_1,\ldots,D_6,D_{12},\ldots,D_{15}\})$ and  $\un(G)$ is a co-trivially perfect graph.
\end{enumerate}
\end{theorem}

%%%%%%%%%%%%%%%%%%%%%%%%%%%%%%%%%%%%%%%%%%%%%%%%%%%%%%%%%%%
\subsection{Oriented trivially perfect graphs}\label{sec-otpg}
%%%%%%%%%%%%%%%%%%%%%%%%%%%%%%%%%%%%%%%%%%%%%%%%%%%%%%%%%%%

\begin{definition}[Oriented trivially perfect graphs] \label{otpg}
The class of {\em oriented trivially perfect graphs} is recursively defined as follows.
\begin{enumerate}[(i)]
\item Every digraph on a single vertex $(\{v\},\emptyset)$,
denoted by $\bullet$, is an oriented trivially perfect  graph.

\item If $G_1$ and $G_2$ are  oriented trivially perfect  graphs,
then $G_1\oplus G_2$  is an oriented trivially perfect  graph.

\item If $G$ is an  oriented trivially perfect graph, then
\begin{inparaenum}[(a)]
\item
$G\oslash \bullet$ and
\item
$\bullet\oslash G$ are  oriented trivially perfect graphs.
\end{inparaenum}
\end{enumerate}
The class of oriented trivially perfect graphs is denoted by $\OTP$.
\end{definition}

The recursive definition of oriented and undirected  trivially perfect graphs lead to the following
observation.

\begin{observation}\label{obs-otp1}
For every oriented trivially perfect graph $G$ the
underlying undirected graph $\un(G)$ is a trivially perfect graph.
\end{observation}

Similar as for oriented co-graphs we obtain a  definition of $\OTP$ by
forbidden induced subdigraphs.

\begin{theorem}\label{char-otop}
Let $G$ be a digraph. The following properties are equivalent:
\begin{enumerate}
	\item\label{ch-otp1} $G$ is an oriented trivially perfect graph.
	
	\item\label{ch-otp2} $G\in \free(\{D_1, D_5, D_8, D_{12}, \overleftrightarrow{K_2}\})$.
    
    \item\label{ch-otp3a} $G\in\free(\{D_1, D_5, \overleftrightarrow{K_2}\})$
		  and $\un(G)\in \free(\{C_4,P_4\})$.
	
	\item\label{ch-otp3} $G\in\free(\{D_1, D_5, \overleftrightarrow{K_2}\})$
		  and $\un(G)$ is a trivially perfect graph.

    \item\label{ch-otp4} $G$ is  transitive and $G\in\free(\{\overleftrightarrow{K_2}, D_8, D_{12}\})$.
\end{enumerate}
\end{theorem}

\begin{proof}
$(\ref{ch-otp1})\Rightarrow(\ref{ch-otp2})$ If $G$ is an oriented trivially perfect  graph, then
$G$ is a directed  trivially perfect graph and $G\in \free(\{D_1, \dots, D_{15}\})$.
Further $G\in \free(\{\overleftrightarrow{K_2}\})$
because of the missing series composition. This leads to $G\in \free(\{D_1, D_5, D_8, D_{12}, \overleftrightarrow{K_2}\})$.

$(\ref{ch-otp2})\Rightarrow(\ref{ch-otp1})$
If $G\in  \free(\{D_1, D_5, D_8, D_{12}, \overleftrightarrow{K_2}\})$, then $G\in \free(\{D_1, \dots, D_{15}\})$ and
is a directed trivially perfect graph.
Since $G\in \free(\{\overleftrightarrow{K_2}\})$
there is no series operation in any
construction of $G$ which implies that $G$ is an oriented  trivially perfect graph.

$(\ref{ch-otp3a}) \Leftrightarrow(\ref{ch-otp3})$  Since $\forb(\TP)=\{C_4,P_4\}$.

$(\ref{ch-otp2})\Rightarrow(\ref{ch-otp3})$ By Observation \ref{obs-otp1}. 

$(\ref{ch-otp3a})\Rightarrow(\ref{ch-otp2})$ By Observation \ref{forb-u-d2}

$(\ref{ch-otp2})\Rightarrow(\ref{ch-otp4})$ By Lemma \ref{le-co-t}
we know that $G$ is transitive. 

$(\ref{ch-otp4})\Rightarrow(\ref{ch-otp2})$ If $G$ is transitive, then  $G$ has no induced $D_1, D_5$.
\end{proof}

\begin{theorem}\label{th-ori-tp}
A graph $G$ is a trivially perfect graph if and only if there exists an orientation $G'$ of  $G$
such that $G'$ is an oriented trivially perfect graph.
\end{theorem}

\begin{proof}
Let $G$ be a trivially perfect graph. Then $G$ is also a comparability graph,
which implies that $G$ has a transitive orientation $G'$. Since $G\in\free(\{C_4,P_4\})$
it follows that $G'\in \free(\{D_8,D_{12}\})$. Further by definition $G'\in\free(\{\overleftrightarrow{K_2}\})$.
By Theorem \ref{char-otop} we know that $G'$ is an oriented trivially perfect graph.

For the reverse direction let $G'$ be  an oriented trivially perfect graph.
Then by Theorem \ref{char-otop} it holds that $G'\in \free(\{D_8,D_{12}\})$.
Since $D_8$ is the only transitive orientation of $P_4$ and $D_{12}$
is the only transitive orientation of $C_4$ it holds that $\un(G)\in \free(\{C_4,P_4\})$.
Thus   $\un(G)$ is a trivially perfect graph.
\end{proof}

%SPO
\begin{observation}
  If $G\in DTP$ then the underlying undirected graph of the symmetric part of $G$ is trivially perfect and the asymmetric part of $G$ is an oriented trivially perfect digraph.
\end{observation}

This holds since the asymmetric part is exactly build with the same rules like trivially perfect graphs and the asymmetric part with the rules of \OTP.

%%%%%%%%%%%%%%%%%%%%%%%%%%%%%%%%%%%%%%%%%%%%%%%%%%%%%%%%%%%
%%%%%%%%%%%%%%%%%%%%%%%%%%%%%%%%%%%%%%%%%%%%%%%%%%%%%%%%%
%%%%%%%%%%%%%%%%%%%%%%%%%%%%%%%%%%%%%%%%%%%%%%%%%%%%%%%%%%
%%%%%%%%%%%%%%%%%%%%%%%%%%%%%%%%%%%%%%%%%%%%%%%%%%%%%%%%%%

\begin{definition}[Oriented co-trivially perfect graphs] \label{octpg}
The class of {\em oriented co-trivially perfect graphs} is recursively defined as follows.
\begin{enumerate}[(i)]
\item Every digraph on a single vertex $(\{v\},\emptyset)$,
denoted by $\bullet$, is an oriented co-trivially perfect  graph.

\item If $G$ is an oriented co-trivially perfect graph,
then
\begin{inparaenum}[(a)]
\item
$G\oslash \bullet$,
\item
$\bullet\oslash G$, and
\item
$G\oplus \bullet$  are oriented co-trivially perfect  graphs.
\end{inparaenum}

\end{enumerate}
The class of oriented co-trivially perfect graphs is denoted by $\OCTP$.
\end{definition}

Restricting the operations of directed co-trivially perfect graphs
to oriented graphs leads to the same operations as
the class of orientated threshold graphs, which will
be considered in Section  \ref{sec-or-tresh}.

%%%%%%%%%%%%%%%%%%%%%%%%%%%%%%%%%%%%%%%%%%%%%%%%%%%%%%%%%%%
\subsection{Directed Weakly Quasi Threshold Graphs}\label{sec-dwqtg}
%%%%%%%%%%%%%%%%%%%%%%%%%%%%%%%%%%%%%%%%%%%%%%%%%%%%%%%%%%%

\begin{definition}[Directed weakly quasi threshold graphs] \label{ddwqt}
The class of {\em directed weakly quasi threshold graphs} is recursively defined as follows.
\begin{enumerate}[(i)]
\item Every edgeless digraph is a directed weakly quasi threshold graph.

\item If $G_1$ and $G_2$ are directed  weakly quasi threshold graphs,
then $G_1\oplus G_2$  is a directed  weakly quasi threshold graph.

\item
If $G$ is a  directed  weakly quasi threshold graph and $I$ is an edgeless digraph,
then
\begin{inparaenum}[(a)]
\item
$G\oslash I$,
\item
$I \oslash G$, and
\item
$G\otimes I$  are  directed  weakly quasi threshold graphs.
\end{inparaenum}
\end{enumerate}
The class of directed weakly quasi threshold graphs is denoted by $\DWQT$.
\end{definition}

\begin{observation}\label{observation:und_dir_WQT}
	If $G$ is a directed weakly quasi threshold graph, $\un(G)$ is weakly quasi threshold graph.
\end{observation}

\begin{table}[h]
  \caption{Digraphs to build the forbidden subdigraphs of $\DWQT$.}\label{tbl:Ymengen}
  \center
  \begin{tabular}{|c|c|c|c|}
  \hline
   & & & \\
	\begin{tikzpicture}
  % How to define variables:
  %  \def \variable {wert};
  % Y_1
  % nodes
  \draw (0,0) \bnode (n1) {};
  \draw (1,0) \bnode (n2) {};

  \draw [-\ahead] (n1) edge (n2);
  \draw [-\ahead] (n2) edge (n1);
\end{tikzpicture} & \begin{tikzpicture}
  % How to define variables:
  %  \def \variable {wert};
  % Y_2
  % nodes
  \draw (0,0) \bnode (n1) {};
  \draw (1,0) \bnode (n2) {};

  \draw [-\ahead] (n1) edge (n2);
\end{tikzpicture} & \begin{tikzpicture}
  % How to define variables:
  %  \def \variable {wert};
  % Y_3
  % nodes
  \draw (0,0) \bnode (n1) {};
  \draw (1,0) \bnode (n2) {};
  \draw (2,0) \bnode (n3) {};

  \draw [-\ahead] (n1) edge (n2);
  \draw [-\ahead] (n2) edge (n1);
\end{tikzpicture}  & \begin{tikzpicture}
  % How to define variables:
  %  \def \variable {wert};
  % Y_4
  % nodes
  \draw (0,0) \bnode (n1) {};
  \draw (1,0) \bnode (n2) {};
  \draw (2,0) \bnode (n3) {};

  \draw [-\ahead] (n1) edge (n2);
\end{tikzpicture} \\
	$Y_1$ & $Y_2$ & $Y_3$ & $Y_4$ \\
	 & & & \\
  \hline
  \end{tabular}
\end{table}

\begin{table}[h]
\caption{Constructions of the forbidden subdigraphs of $\DWQT$.}\label{table:Combinations_Ymengen}
  \center
  \begin{tabular}{|c|c|c|c|}
  \hline
	$Q_1 = Y_2 \otimes Y_2$ & $Q_2 = Y_1 \oslash Y_1$   & $Q_3 = Y_3 \otimes Y_3$   & $Q_4 = Y_2 \otimes Y_3$ \\
	\hline
    $Q_5 = Y_1 \oslash Y_4$ &   $Q_6 = Y_4 \oslash Y_1$ &   $Q_7 = Y_4 \oslash Y_4$ & \\
  \hline
  \end{tabular}
\end{table}

\begin{theorem}\label{char-dwqt}
Let $G$ be a digraph. The following properties are equivalent:
\begin{enumerate}
	\item\label{ch-dwqt1}
  $G$ is a directed weakly quasi threshold graph.
	\item\label{ch-dwqt2}
  $G\in \free(\{D_1,\ldots,D_8,Q_1, \ldots, Q_7\})$\footnote{Note that $Q_1=D_{11}$ and $Q_2 = D_{15}$.}.
  \item\label{ch-dwqt3}
  $G\in \free(\{D_1,\ldots,D_6,Q_1, Q_2, Q_4, Q_5, Q_6\})$ and $\un(G)\in \free(\{P_4,  \co 2P_3 \})$.
  \item\label{ch-dwqt4}
  $G\in \free(\{D_1,\ldots,D_6,Q_1, Q_2, Q_4, Q_5, Q_6\})$ and $\un(G)$ is a weakly quasi threshold graph.
\end{enumerate}
\end{theorem}

\begin{proof}
  $(\ref{ch-dwqt1}) \Rightarrow (\ref{ch-dwqt2})$
  The given forbidden digraphs  $D_1,\ldots,D_8,Q_1, \ldots, Q_7$ are not directed weakly quasi threshold graphs and the set of directed weakly quasi threshold graphs is hereditary.

 $(\ref{ch-dwqt2}) \Rightarrow (\ref{ch-dwqt1})$
 Let $G$ be a digraph without induced $D_1,\ldots,D_8,Q_1, \ldots, Q_7$. Since there are no induced $D_1,\ldots,D_8$, it holds that $G\in \DC$. Thus, $G$ is constructed by the disjoint union, the series and the order composition.
 $Q_1, Q_3$ and $ Q_4$ can only be build by a series composition of two graphs $G_1$ and $G_2$, where $G_1,G_2\in \{Y_2,Y_3\}$.
 We note that $Y_2, Y_3$ are contained in every directed co-graph containing more vertices, that is not a sequence of length at least one of series compositions of independent sets.
 Consequently, there are no bigger forbidden induced subdigraphs that emerged through a series operation, such that the $Q_1, Q_3$ and $ Q_4$ characterize exactly the allowed series compositions in $\DWQT$.

 The $Q_2, Q_5, Q_6$ and $ Q_7$ can only be build by an order composition of two graphs $G_1$ and $G_2$, where $G_1,G_2\in \{Y_1,Y_4\}$. We note that $Y_1, Y_4$ are contained in every directed co-graph containing more vertices, that is not a sequence of length at least one of order compositions of independent sets.
 Consequently, there are no forbidden induced subdigraphs containing more vertices that emerged through an order operation, such that the $Q_2, Q_5, Q_6$ and $ Q_7$ characterize exactly the allowed order compositions in $\DWQT$.
 Finally, by excluding these digraphs, we end up in the Definition \ref{ddwqt} for $\DWQT$, such that $G\in \DWQT$.

$(\ref{ch-dwqt2}) \Rightarrow (\ref{ch-dwqt4})$  By Observation \ref{observation:und_dir_WQT}.

  $(\ref{ch-dwqt3}) \Rightarrow (\ref{ch-dwqt2})$
  By Observation \ref{forb-u-d2}.
  
	$(\ref{ch-dwqt3}) \Leftrightarrow (\ref{ch-dwqt4})$
  Since $\forb(\WQT)= \{P_4,  \co 2P_3\}$.
\end{proof}

\begin{table}[h]
  \caption{Forbidden induced subdigraphs for $\DWQT$.}\label{tbl:DWQT}
  \center
  \begin{tabular}{c c c c }
	\begin{tikzpicture}
  % How to define variables:
  %  \def \variable {wert};
  % DWQT_1
  % nodes
  \draw (0,0) \bnode (n1) {};
  \draw (1,0) \bnode (n2) {};
  \draw (0,1) \bnode (n3) {};
  \draw (1,1) \bnode (n4) {};
  % edges
  \draw [-\ahead] (n1) edge (n2);
  \draw [-\ahead] (n3) edge (n4);

  % doubleedges
  \draw [-\ahead] (n1) edge (n3);
  \draw [-\ahead] (n1) edge (n4);
  \draw [-\ahead] (n2) edge (n3);
  \draw [-\ahead] (n2) edge (n4);
  \draw [-\ahead] (n3) edge (n1);
  \draw [-\ahead] (n3) edge (n2);
  \draw [-\ahead] (n4) edge (n1);
  \draw [-\ahead] (n4) edge (n2);
\end{tikzpicture} & \begin{tikzpicture}
  % How to define variables:
  %  \def \variable {wert};
  % nodes
  \draw (0,0) \bnode (n1) {};
  \draw (1,0) \bnode (n2) {};
  \draw (0,1) \bnode (n3) {};
  \draw (1,1) \bnode (n4) {};
  % edges
  \draw [-\ahead] (n1) edge (n2);
  \draw [-\ahead] (n3) edge (n4);

  % doubleedges
  \draw [-\ahead] (n1) edge (n2);
  \draw [-\ahead] (n3) edge (n4);
  \draw [-\ahead] (n2) edge (n1);
  \draw [-\ahead] (n4) edge (n3);
  
  \draw [-\ahead] (n3) edge (n1);
  \draw [-\ahead] (n3) edge (n2);
  \draw [-\ahead] (n4) edge (n1);
  \draw [-\ahead] (n4) edge (n2);
\end{tikzpicture} & \begin{tikzpicture}
  % How to define variables:
  %  \def \variable {wert};
  % nodes
  \draw (0,1) \bnode (n0) {};
  \draw (1,1) \bnode (n1) {};
  \draw (2,1) \bnode (n2) {};
  \draw (0,0) \bnode (n3) {};
  \draw (1,0) \bnode (n4) {};
  \draw (2,0) \bnode (n5) {};
  % edges
  \draw [-\ahead] (n0) edge (n1);
  \draw [-\ahead] (n1) edge (n0);

  \draw [-\ahead] (n3) edge (n4);
  \draw [-\ahead] (n4) edge (n3);
  % double edges
  \draw [-\ahead] (n0) edge (n3);
  \draw [-\ahead] (n1) edge (n3);
  \draw [-\ahead] (n2) edge (n3);
  \draw [-\ahead] (n0) edge (n4);
  \draw [-\ahead] (n1) edge (n4);
  \draw [-\ahead] (n2) edge (n4);
  \draw [-\ahead] (n0) edge (n5);
  \draw [-\ahead] (n1) edge (n5);
  \draw [-\ahead] (n2) edge (n5);

  \draw [-\ahead] (n3) edge (n0);
  \draw [-\ahead] (n3) edge (n1);
  \draw [-\ahead] (n3) edge (n2);
  \draw [-\ahead] (n4) edge (n0);
  \draw [-\ahead] (n4) edge (n1);
  \draw [-\ahead] (n4) edge (n2);
  \draw [-\ahead] (n5) edge (n0);
  \draw [-\ahead] (n5) edge (n2);
  \draw [-\ahead] (n5) edge (n1);

\end{tikzpicture} & \begin{tikzpicture}
  % How to define variables:
  %  \def \variable {wert};
  % nodes
  \draw (0.5,1) \bnode (n1) {};
  \draw (1.5,1) \bnode (n2) {};
  \draw (0,0) \bnode (n3) {};
  \draw (1,0) \bnode (n4) {};
  \draw (2,0) \bnode (n5) {};
  % edges
  \draw [-\ahead] (n1) edge (n2);

  \draw [-\ahead] (n3) edge (n4);
  \draw [-\ahead] (n4) edge (n3);

  % doubleedges
  \draw [-\ahead] (n1) edge (n3);
  \draw [-\ahead] (n1) edge (n4);
  \draw [-\ahead] (n1) edge (n5);
  \draw [-\ahead] (n2) edge (n3);
  \draw [-\ahead] (n2) edge (n4);
  \draw [-\ahead] (n2) edge (n5);

  \draw [-\ahead] (n3) edge (n1);
  \draw [-\ahead] (n4) edge (n1);
  \draw [-\ahead] (n5) edge (n1);
  \draw [-\ahead] (n3) edge (n2);
  \draw [-\ahead] (n4) edge (n2);
  \draw [-\ahead] (n5) edge (n2);
\end{tikzpicture} \\
	$Q_1$ & $Q_2$ & $Q_3$ & $Q_4$ \\
	 & & & \\
	\begin{tikzpicture}
  % V_13
  % How to define variables:
  %  \def \variable {wert};
  % nodes
  \draw (1,1.4) \bnode (n1) {};
  \draw (0,0.7) \bnode (n2) {};
  \draw (1,0.7) \bnode (n3) {};
  \draw (2,0.7) \bnode (n4) {};
  \draw (1,0) \bnode (n5) {};
  % edges
  % simples
  \draw [-\ahead] (n2) edge (n1);
  \draw [-\ahead] (n2) edge (n3);
  \draw [-\ahead] (n2) edge (n5);
  \draw [-\ahead] (n4) edge (n1);
  \draw [-\ahead] (n4) edge (n3);
  \draw [-\ahead] (n4) edge (n5);

  \draw [-\ahead,in=0,out=0, looseness=2.8] (n5) edge (n1);
  \draw [-\ahead,in=90,out=90, looseness=1.5] (n2) edge (n4);
  \draw [-\ahead,in=90,out=90, looseness=1.5] (n4) edge (n2);
\end{tikzpicture} & \begin{tikzpicture}
  % V_13
  % How to define variables:
  %  \def \variable {wert};
  % nodes
  \draw (1,1.4) \bnode (n1) {};
  \draw (0,0.7) \bnode (n2) {};
  \draw (1,0.7) \bnode (n3) {};
  \draw (2,0.7) \bnode (n4) {};
  \draw (1,0) \bnode (n5) {};
  % edges
  % simples
  \draw [-\ahead] (n1) edge (n2);
  \draw [-\ahead] (n3) edge (n2);
  \draw [-\ahead] (n5) edge (n2);
  \draw [-\ahead] (n1) edge (n4);
  \draw [-\ahead] (n3) edge (n4);
  \draw [-\ahead] (n5) edge (n4);

  \draw [-\ahead,in=0,out=0, looseness=2.8] (n5) edge (n1);
  \draw [-\ahead,in=90,out=90, looseness=1.5] (n2) edge (n4);
  \draw [-\ahead,in=90,out=90, looseness=1.5] (n4) edge (n2);
\end{tikzpicture}  & \begin{tikzpicture}
  % DWQT_7
  % How to define variables:
  %  \def \variable {wert};
  % nodes
  \draw (1,1) \bnode (n1) {};
  \draw (2,1) \bnode (n2) {};
  \draw (3,1) \bnode (n3) {};
  \draw (1,2) \bnode (n4) {};
  \draw (2,2) \bnode (n5) {};
  \draw (3,2) \bnode (n6) {};
  % edges
  % simples
  \draw [-\ahead] (n4) edge (n1);
  \draw [-\ahead] (n4) edge (n2);
  \draw [-\ahead] (n4) edge (n3);
  \draw [-\ahead] (n5) edge (n1);
  \draw [-\ahead] (n5) edge (n2);
  \draw [-\ahead] (n5) edge (n3);
  \draw [-\ahead] (n6) edge (n1);
  \draw [-\ahead] (n6) edge (n2);
  \draw [-\ahead] (n6) edge (n3);

  \draw [-\ahead,in=130,out=50, looseness=0.7] (n4) edge (n6);
  \draw [-\ahead,in=-130,out=-50, looseness=0.7] (n1) edge (n3);
\end{tikzpicture} & \\
	$Q_5$ & $Q_6$ & $Q_7$ &\\
  \end{tabular}
\end{table}

Since $\co \{Q_1,\ldots,Q_7\}\neq \{Q_1,\ldots,Q_7\}$ it is reasonable to introduce the complementary class of $\DWQT$. 
%Therefore, we need the definition of \emph{directed cliques}.
%A directed clique is a bidirectional complete digraph, such that $K=(V,E)$ with $E=\{(u,v) \mid \forall u,v %\in V,  u \neq v \}$.

%%%%%%%%%%%%%%%%%%%%%%%%%%%%%%%%%%%%%%%%%%%%%%%%%%%%%%%%%%%
\subsection{Oriented weakly quasi threshold graphs}\label{sec-owqtg}
%%%%%%%%%%%%%%%%%%%%%%%%%%%%%%%%%%%%%%%%%%%%%%%%%%%%%%%%%%%

\begin{definition}[Oriented weakly quasi threshold graphs] \label{dowqt}
The class of {\em oriented weakly quasi threshold graphs} is recursively defined as follows.
\begin{enumerate}[(i)]
\item Every edgeless digraph is an oriented weakly quasi threshold graph.

\item If $G_1$ and $G_2$ are oriented  weakly quasi threshold graphs,
then $G_1\oplus G_2$  is an oriented  weakly quasi threshold graph.

\item
If $G$ is an oriented weakly quasi threshold   graph and $I$ is an edgeless digraph,
then
\begin{inparaenum}[(a)]
\item
$G\oslash I$ and
\item
$I \oslash G$ are oriented  weakly quasi threshold graphs.
\end{inparaenum}
\end{enumerate}
The class of oriented weakly quasi threshold graphs is denoted by $\OWQT$.
\end{definition}

\begin{theorem}\label{char-owqt}
Let $G$ be a digraph. The following properties are equivalent:
\begin{enumerate}
	\item\label{cc-owqt1}
  $G$ is a oriented  weakly quasi threshold graph.
	\item\label{cc-owqt2}
  $G \in \free(\{D_1,D_5,D_8, \overleftrightarrow{K_2},Q_7\})$.
  \item\label{cc-owqt3}
  $G \in \free(\{D_8, \overleftrightarrow{K_2}, Q_7 \})$ and $G$ is transitive.
  \item\label{cc-owqt4}
  $G \in \free(\{D_1,D_5,\overleftrightarrow{K_2} \})$ and $\un(G)$ is a weakly quasi threshold graph.
\end{enumerate}
\end{theorem}

\begin{proof}
	$(\ref{cc-owqt1})\Rightarrow (\ref{cc-owqt2})$
	Let $G$ be in $\OWQT$. Since $\OWQT \subset \OC$ we know that $D_1,D_5, D_8$ and $\overleftrightarrow{K_2}$ are forbidden.
	Further, since  $\OWQT \subset \DWQT$  the graphs of Table \ref{tbl:DWQT} are forbidden.
	As $\overleftrightarrow{K_2}$ is forbidden, too, this reduces the table and only $Q_7$ remains.
	Since $\OWQT$ is closed under taking induced subdigraphs, this holds for every digraph in this class.

	$(\ref{cc-owqt2})\Rightarrow (\ref{cc-owqt1})$
	Since $D_1, D_5, D_8$ and $\overleftrightarrow{K_2}$ are forbidden induced subdigraphs of $G$, it holds that $G$ is an oriented co-graph, see Theorem \ref{char-oc}.
	Consequently, $G$ has been constructed by the operations allowed for oriented co-graphs, which are the disjoint union and the order composition.
	Further, $Q_7$ is excluded, since it must have been the result of an order composition.
	Since $Q_7=Y_4 \oslash Y_4$ and $Y_4$ is neither an edgeless digraph nor a transitive tournament, this order composition is not valid in $\OWQT$ and thus, $Q_7$ is a forbidden subdigraph.
	For every other possible order composition of two oriented digraphs $G_1$ and $G_2$, either $G_1$ or $G_2$ is edgeless or a transitive tournament, which is allowed in $\OWQT$, or both graphs contain $Y_4$ as induced subdigraph, such that the forbidden $Q_7$ completely characterizes the restriction of the order compositions in $\OWQT$. Since $G$ does not contain one of these forbidden subdigraphs, $G\in OWQT$ follows.
	
	$(\ref{cc-owqt2})\Rightarrow (\ref{cc-owqt3})$
	Follows by Lemma \ref{le-tt}.

	$(\ref{cc-owqt3}) \Rightarrow (\ref{cc-owqt2})$
	The transitivity of $G$ implies that $D_1$ and $D_5$ are forbidden, as they do not satisfy the definition of transitivity.

  $(\ref{cc-owqt2}) \Leftrightarrow (\ref{cc-owqt4})$ Since $\forb(\WQT)=\{P_4, \co 2P_3\}$.
\end{proof}

\begin{observation}
  If $G\in \DWQT$ then the underlying undirected graph of the symmetric part of $G$ is a weakly quasi threshold graph and the asymmetric part of $G$ is an oriented weakly quasi threshold graph.
\end{observation}

This holds since the asymmetric part is exactly build with the same rules like weakly quasi threshold graphs and the asymmetric part with the rules of \OWQT.

%%%%%%%%%%%%%%%%%%%%%%%%%%%%%%%%%%%%%%%%%%%%%%%%%%%%%%%%%%%
\subsection{Directed co-weakly quasi threshold graphs}\label{sec-dcwqtg}
%%%%%%%%%%%%%%%%%%%%%%%%%%%%%%%%%%%%%%%%%%%%%%%%%%%%%%%%%%%

\begin{definition}[Directed co-weakly quasi threshold graphs] \label{dcwqt}
The class of {\em directed co-weakly quasi threshold graphs} is recursively defined as follows.
\begin{enumerate}[(i)]
\item Every bidirectional complete digraph is a directed co-weakly quasi threshold graph.

\item If $G_1$ and $G_2$ are directed  co-weakly quasi threshold graphs,
then $G_1\otimes G_2$  is a directed  co-weakly quasi threshold graph.

\item
If $G$ is a directed  co-weakly quasi threshold graph and $K$ is a  bidirectional complete digraph,
then
\begin{inparaenum}[(a)]
\item
$G\oslash K$,
\item
$K \oslash G$, and
\item
$G\oplus K$  are  directed  co-weakly quasi threshold graphs.
\end{inparaenum}
\end{enumerate}
The class of directed co-weakly quasi threshold graphs is denoted by $\DCWQT$.
\end{definition}

\begin{theorem}\label{char-dcwqt}
	For a graph $G$ the following properties are equivalent.
	\begin{enumerate}
		\item $G \in \DCWQT$.
		\item $G \in \free(\{D_1,\ldots, D_8, \co Q_1, \ldots, \co Q_7 \})$, see Table \ref{tbl:DCWQT}.
	\end{enumerate}
\end{theorem}

\begin{proof}
	By Lemma \ref{th-free}, we know that $\co \DWQT=\free(\co \{D_1,\ldots, D_8, Q_1, \ldots, Q_7 \})$. Since $\DCWQT=\co \DWQT$,	it follows by Theorem \ref{char-dwqt} that the forbidden induced subdigraphs are $D_1,\ldots, D_8$ since they are self complementary, and $\co Q_1, \ldots, \co Q_7$. 
\end{proof}

\begin{table}[h]
  \caption{Forbidden induced subdigraphs for $\DCWQT$.}\label{tbl:DCWQT}
  \center

  \begin{tabular}{c c c c }
	\begin{tikzpicture}
  % How to define variables:
  %  \def \variable {wert};
  % nodes
  \draw (0,0) \bnode (n1) {};
  \draw (1,0) \bnode (n2) {};
  \draw (0,1) \bnode (n3) {};
  \draw (1,1) \bnode (n4) {};
  % edges
  \draw [-\ahead] (n1) edge (n2);
  \draw [-\ahead] (n3) edge (n4);
\end{tikzpicture} & \begin{tikzpicture}
  % How to define variables:
  %  \def \variable {wert};
  % nodes
  \draw (0.5,1) \bnode (n1) {};
  \draw (1.5,1) \bnode (n2) {};
  \draw (0,0) \bnode (n3) {};
  \draw (1,0) \bnode (n4) {};
  \draw (2,0) \bnode (n5) {};
  % edges
  \draw [-\ahead] (n1) edge (n2);
  \draw [-\ahead] (n3) edge (n4);
  \draw [-\ahead] (n4) edge (n3);
  \draw [-\ahead] (n4) edge (n5);
  \draw [-\ahead] (n5) edge (n4);
\end{tikzpicture} & \begin{tikzpicture}
  % How to define variables:
  %  \def \variable {wert};
  % nodes
  \draw (0,1) \bnode (n0) {};
  \draw (1,1) \bnode (n1) {};
  \draw (2,1) \bnode (n2) {};
  \draw (0,0) \bnode (n3) {};
  \draw (1,0) \bnode (n4) {};
  \draw (2,0) \bnode (n5) {};
  % edges
  \draw [-\ahead] (n0) edge (n1);
  \draw [-\ahead] (n1) edge (n0);
  \draw [-\ahead] (n1) edge (n2);
  \draw [-\ahead] (n2) edge (n1);

  \draw [-\ahead] (n3) edge (n4);
  \draw [-\ahead] (n4) edge (n3);
  \draw [-\ahead] (n4) edge (n5);
  \draw [-\ahead] (n5) edge (n4);
\end{tikzpicture} & \begin{tikzpicture}
  % How to define variables:
  %  \def \variable {wert};
  % nodes
  \draw (0,0) \bnode (n1) {};
  \draw (1,0) \bnode (n2) {};
  \draw (0,1) \bnode (n3) {};
  \draw (1,1) \bnode (n4) {};
  % edges
  \draw [-\ahead] (n1) edge (n2);
  \draw [-\ahead] (n1) edge (n3);
  \draw [-\ahead] (n4) edge (n2);
  \draw [-\ahead] (n4) edge (n3);
\end{tikzpicture} \\
	$\co Q_1$ & $\co Q_4$ & $\co Q_3$ & $\co Q_2$ \\
	 & & & \\
	\begin{tikzpicture}
  % V_14
  % How to define variables:
  %  \def \variable {wert};
  % nodes
  \draw (1,2) \bnode (n1) {};
  \draw (0,1) \bnode (n2) {};
  \draw (1,1) \bnode (n3) {};
  \draw (2,1) \bnode (n4) {};
  \draw (1,0) \bnode (n5) {};
  % edges
  % doubles
  \draw [-\ahead] (n1) edge (n3);
  \draw [-\ahead] (n3) edge (n1);
  \draw [-\ahead] (n5) edge (n3);
  \draw [-\ahead] (n3) edge (n5);
  % simples
  \draw [-\ahead] (n1) edge (n2);
  \draw [-\ahead] (n3) edge (n2);
  \draw [-\ahead] (n5) edge (n2);
  \draw [-\ahead] (n1) edge (n4);
  \draw [-\ahead] (n3) edge (n4);
  \draw [-\ahead] (n5) edge (n4);

  \draw [-\ahead,in=0,out=0, looseness=2] (n5) edge (n1);
\end{tikzpicture} & \begin{tikzpicture}
  % V_13
  % How to define variables:
  %  \def \variable {wert};
  % nodes
  \draw (1,2) \bnode (n1) {};
  \draw (0,1) \bnode (n2) {};
  \draw (1,1) \bnode (n3) {};
  \draw (2,1) \bnode (n4) {};
  \draw (1,0) \bnode (n5) {};
  % edges
  % doubles
  \draw [-\ahead] (n1) edge (n3);
  \draw [-\ahead] (n3) edge (n1);
  \draw [-\ahead] (n5) edge (n3);
  \draw [-\ahead] (n3) edge (n5);
  % simples
  \draw [-\ahead] (n2) edge (n1);
  \draw [-\ahead] (n2) edge (n3);
  \draw [-\ahead] (n2) edge (n5);
  \draw [-\ahead] (n4) edge (n1);
  \draw [-\ahead] (n4) edge (n3);
  \draw [-\ahead] (n4) edge (n5);

  \draw [-\ahead,in=0,out=0, looseness=2] (n5) edge (n1);
\end{tikzpicture}  & \begin{tikzpicture}
  % V_19
  % How to define variables:
  %  \def \variable {wert};
  % nodes
  \draw (1,1) \bnode (n1) {};
  \draw (2,1) \bnode (n2) {};
  \draw (3,1) \bnode (n3) {};
  \draw (1,2) \bnode (n4) {};
  \draw (2,2) \bnode (n5) {};
  \draw (3,2) \bnode (n6) {};
  % edges
  % doubles
  \draw [-\ahead] (n2) edge (n1);
  \draw [-\ahead] (n1) edge (n2);
  \draw [-\ahead] (n3) edge (n2);
  \draw [-\ahead] (n2) edge (n3);
  \draw [-\ahead] (n4) edge (n5);
  \draw [-\ahead] (n5) edge (n4);
  \draw [-\ahead] (n5) edge (n6);
  \draw [-\ahead] (n6) edge (n5);
  % simples
  \draw [-\ahead] (n4) edge (n1);
  \draw [-\ahead] (n4) edge (n2);
  \draw [-\ahead] (n4) edge (n3);
  \draw [-\ahead] (n5) edge (n1);
  \draw [-\ahead] (n5) edge (n2);
  \draw [-\ahead] (n5) edge (n3);
  \draw [-\ahead] (n6) edge (n1);
  \draw [-\ahead] (n6) edge (n2);
  \draw [-\ahead] (n6) edge (n3);

  \draw [-\ahead,in=120,out=60, looseness=1] (n4) edge (n6);
  \draw [-\ahead,in=-120,out=-60, looseness=1] (n1) edge (n3);
\end{tikzpicture} & \\
	$\co Q_5$ & $\co Q_6$ & $\co Q_7$ &\\
  \end{tabular}
\end{table}

%%%%%%%%%%%%%%%%%%%%%%%%%%%%%%%%%%%%%%%%%%%%%%%%%%%%%%%%%%%
\subsection{Oriented co-weakly quasi threshold graphs}\label{sec-ocwqtg}
%%%%%%%%%%%%%%%%%%%%%%%%%%%%%%%%%%%%%%%%%%%%%%%%%%%%%%%%%%%

\begin{definition}[Oriented co-weakly quasi threshold graphs] \label{ocwqt}
The class of {\em oriented co-weakly quasi threshold graphs} is recursively defined as follows.
\begin{enumerate}[(i)]
\item Every transitive tournament is an oriented co-weakly quasi threshold graph.

\item
If $G$ is an oriented co-weakly quasi threshold graph and $T$ is a transitive tournament,
then
\begin{inparaenum}[(a)]
\item
$G\oslash T$,
\item
$T \oslash G$, and
\item
$G\oplus T$  are  oriented co-weakly quasi threshold graphs.
\end{inparaenum}
\end{enumerate}
The class of oriented co-weakly quasi threshold graphs is denoted by $\OCWQT$.
\end{definition}

Obviously, a transitive tournament $\overrightarrow{T_n}$ is a subdigraph, but not an induced subdigraph of $\overleftrightarrow{K_n}$.
Thus, the operations allowed in $\OCWQT$ are not exactly building the complement of the digraphs in $\DCWQT$, but since biorientations are forbidden in oriented digraphs, it is sensible to define the class like we did above.

\begin{table}\label{fig:forbiddenOCWQT}
	\begin{center}
	\caption{Forbidden induced subdigraphs of $\OCWQT$.}
  	\vspace{5mm}
		\begin{tikzpicture}
  % D_21, D_22, D_23
  % How to define variables:
  %  \def \variable {wert};
  % D_21
  % nodes
  \draw (0,0.5) \bnode (n1) {};
  \draw (1,0.5) \bnode (n2) {};
  \draw (2,0.5) \bnode (n3) {};
  \draw (0,1) \bnode (n4) {};
  \draw (1,1) \bnode (n5) {};
  \draw (2,1) \bnode (n6) {};
  %edges
  \draw [-\ahead] (n1) edge (n2);
  \draw [-\ahead] (n3) edge (n2);
  \draw [-\ahead] (n4) edge (n5);
  \draw [-\ahead] (n6) edge (n5);

  \draw (1,0) node (t1){$D_{21}$};

  % D_22
  % nodes
  \draw (3,0.5) \bnode (n1) {};
  \draw (4,0.5) \bnode (n2) {};
  \draw (5,0.5) \bnode (n3) {};
  \draw (3,1) \bnode (n4) {};
  \draw (4,1) \bnode (n5) {};
  \draw (5,1) \bnode (n6) {};
  %edges
  \draw [-\ahead] (n2) edge (n1);
  \draw [-\ahead] (n2) edge (n3);
  \draw [-\ahead] (n5) edge (n4);
  \draw [-\ahead] (n5) edge (n6);

  \draw (4,0) node (t2){$D_{22}$};

  % D_23
  % nodes
  \draw (6,0.5) \bnode (n1) {};
  \draw (7,0.5) \bnode (n2) {};
  \draw (8,0.5) \bnode (n3) {};
  \draw (6,1) \bnode (n4) {};
  \draw (7,1) \bnode (n5) {};
  \draw (8,1) \bnode (n6) {};
  %edges
  \draw [-\ahead] (n2) edge (n1);
  \draw [-\ahead] (n2) edge (n3);
  \draw [-\ahead] (n4) edge (n5);
  \draw [-\ahead] (n6) edge (n5);

  \draw (7,0) node (t3){$D_{23}$};

\end{tikzpicture}
	\end{center}
\end{table}

\begin{table}[ht]
\caption{Digraphs to build the forbidden subdigraphs of $\OCWQT$.}
\label{fig:X_of_OCWQT}
	\begin{center}
		\begin{tikzpicture}
  % D_21, D_22, D_23
  % How to define variables:
  %  \def \variable {wert};
  % nodes
  \draw (0,1) \bnode (n1) {};
  \draw (1,1) \bnode (n2) {};
  \draw (2,1) \bnode (n3) {};

  %edges
  \draw [-\ahead] (n1) edge (n2);
  \draw [-\ahead] (n3) edge (n2);

  \draw (1,0.5) node (t1){$X_1$};

  % nodes
  \draw (3,1) \bnode (n1) {};
  \draw (4,1) \bnode (n2) {};
  \draw (5,1) \bnode (n3) {};

  %edges
  \draw [-\ahead] (n2) edge (n1);
  \draw [-\ahead] (n2) edge (n3);

  \draw (4,0.5) node (t2){$X_2$};

\end{tikzpicture}
	\end{center}
\end{table}

\begin{theorem}\label{char-ocwqt}
	For a graph $G$ the following properties are equivalent.
	\begin{enumerate}
     \item\label{occc1}
    $G$ is a oriented co-weakly quasi threshold graph.
		\item\label{occc2}
    $G \in \free(\{D_1, D_5, D_8, \overleftrightarrow{K_2}, D_{12}, D_{21}, D_{22}, D_{23}\})$\footnote{$D_{12}$ is equal to $\co Q_2$.}.
    \item\label{occc3}
    $G$ is an oriented co-graph and $G \in \free(\{D_{12}, D_{21}, D_{22}, D_{23}\})$.
    \item\label{occc4}
    $G \in \free(\{D_8, \overleftrightarrow{K_2}, D_{12}, D_{21}, D_{22}, D_{23}\})$ and $G$ is transitive.
	\end{enumerate}
\end{theorem}

\begin{proof}
	$(\ref{occc1}) \Rightarrow (\ref{occc2})$
	Let $G\in \OCWQT$ and assume $G$ contains one of  $D_1, D_5, D_8, \overleftrightarrow{K_2}, D_{12}, D_{21}, D_{22}, D_{23}$.
	Since these digraphs are not constructed by the operations of $\OCWQT$, this contradicts that $G$ is in this class. Since $\OCWQT$ is closed under taking induced subdigraphs, $G$ cannot contain one of these graphs as induced subdigraphs.

	$(\ref{occc2}) \Rightarrow (\ref{occc1})$
	Since $D_1, D_5, D_8$ and $\overleftrightarrow{K_2}$ are forbidden induced subdigraphs of $G$, $G$ is an oriented co-graph, see Theorem \ref{char-ocwqt}. Thus, $G$ has been constructed by the operations allowed for oriented co-graphs, which are the disjoint union and the order composition.

	Since $D_{12}$ is not an induced subdigraph of $G$, in every order composition of two graphs $G_1$ and $G_2$ in the construction of $G$, at least one of them must have been a transitive tournament. For every bigger digraph, composed by an order composition of two oriented digraphs $G_1$ and $G_2$, it holds that either one of them is a transitive tournament, which is allowed in the order composition in $\OCWQT$, or else $I_2$ is included as induced subdigraph in $G_1$ and $G_2$, which means that $D_{12}$ is contained as induced subdigraph. This leads to the conclusion that prohibiting $D_{12}$ exactly characterizes the order composition allowed in $\OCWQT$.

	Since $D_{21}, D_{22}$ and $D_{23}$ are not completely connected, they must be the result of the disjoint union of two oriented co-graphs $G_1$ and $G_2$, with $G_1,G_2\in \{X_1,X_2\}$.
	As $X_1$ and $X_2$ are neither transitive tournaments nor edgeless digraphs, they are forbidden in $\OCWQT$.
	Every bigger oriented digraph emerging by a disjoint union either contains $D_{21}, D_{22}$ and $D_{23}$, or is feasible to build by adding transitive tournaments one by one. Consequently, excluding  $D_{21}, D_{22}$ and $D_{23}$ exactly characterizes the restriction of the disjoint union in $\OCWQT$.

	This exactly leads to Definition \ref{char-owqt} of $\OWQT$, such that $G\in \OCWQT$.

	$(\ref{occc2}) \Leftrightarrow (\ref{occc3})$
	Holds by Theorem \ref{char-oc}.

	$(\ref{occc2}) \Rightarrow (\ref{occc4})$
	Holds by Lemma \ref{le-tt}.

	$(\ref{occc4}) \Rightarrow (\ref{occc2})$
	The transitivity of $G$ implies that there cannot be $D_1$ or $D_5$, as both do not satisfy the definition of transitivity.
\end{proof}

The next class we consider is the class of simple co-graphs, of which we construct a directed version.

%%%%%%%%%%%%%%%%%%%%%%%%%%%%%%%%%%%%%%%%%%%%%%%%%%%%%%%%%%%
\subsection{Directed simple co-graphs}\label{sec-dscg}
%%%%%%%%%%%%%%%%%%%%%%%%%%%%%%%%%%%%%%%%%%%%%%%%%%%%%%%%%%%

\begin{definition}[Directed simple co-graphs] \label{dscg}
The class of {\em directed simple co-graphs} is recursively defined as follows.
\begin{enumerate}[(i)]
\item Every digraph on a single vertex $(\{v\},\emptyset)$,
denoted by $\bullet$, is a directed simple co-graph.

\item
If $G$ is a  directed simple co-graph $I$ is an edgeless digraph,
then
\begin{inparaenum}[(a)]
\item
$G\oplus I$,
\item
$G\oslash I$,
\item
$I \oslash G$, and
\item
$G\otimes I$  are directed simple co-graphs.
\end{inparaenum}
\end{enumerate}
The class of directed simple co-graphs is denoted by $\DSC$.
\end{definition}

\begin{observation}\label{observation:und_dir_SC}
	If $G$ is a directed simple co-graph, $\un(G)$ is a simple co-graph.
\end{observation}

Since directed  simple co-graphs are a subset of directed weakly quasi threshold
graphs they can be defined by adding further subdigraphs to those given
for directed weakly quasi threshold graphs. These have to ensure that for every disjoint union of $G_1$ and $G_2$ either $G_1$
or $G_2$ has no edge. This can be done by
excluding $\co D_{11},\co D_{10},\co D_9$.

\begin{theorem}\label{char-dsc}
	The following statements are equivalent.
	\begin{enumerate}
		\item\label{dsc1} $G\in \DSC$.
		\item\label{dsc2} $G\in \free(\{D_1,\ldots,D_8, Q_1,\ldots,Q_7, \co D_9, \co D_{10}, \co D_{11}\})$.
		\item\label{dsc3} $G \in \free(\{D_1,\ldots,D_8, Q_1,\ldots,Q_7\})$ and $\un(G)\in \free(\{P_4, \co 2P_3, 2K_2\})$.
		\item\label{dsc4} $G \in \free(\{D_1,\ldots,D_8, Q_1,\ldots,Q_7\})$ and $\un(G)\in \SC$.
	\end{enumerate}
\end{theorem}

\begin{proof}
	$(\ref{dsc1}) \Rightarrow (\ref{dsc2})$
    The given forbidden digraphs  $D_1,\ldots,D_8, Q_1,\ldots,Q_7, \co D_9, \co D_{10}, \co D_{11}$ are	no directed simple co-graphs and the set of directed directed simple co-graphs is hereditary.
  
	$(\ref{dsc2}) \Rightarrow (\ref{dsc1})$
	Let $G\in \free(\{D_1,\ldots,D_8, Q_1,\ldots,Q_7, \co D_9, \co D_{10}, \co D_{11}\})$. Then, we know by Theorem \ref{char-dwqt} that $G$ is in $\DWQT$, such that in every series and order composition of two digraphs $G_1$ and $G_2$, either $G_1$ or $G_2$ must be an edgeless digraph (see proof of Theorem \ref{char-dwqt}), exactly as in the definition of $\DSC$.
	By excluding $\co D_9, \co D_{10}$ and $\co D_{11}$ there is no disjoint union $G_1 \oplus G_2$ allowed, in which at least one of the two digraphs contains an edge. Therefore, one of $G_1$ or $G_2$ must be edgeless, which exactly is the restriction for the disjoint union from Definition \ref{dscg}. For every other digraph built via $G_1 \oplus G_2$, either $G_1$ and $G_1$ contain edges,
	such that $\co D_9, \co D_{10}$ or $ \co D_{11}$ is an induced subdigraph, or  $G_1$ or $G_1$ contains no edges, which leads to a legit digraph of $\DSC$.
	Finally,we end up in Definition \ref{dscg}. Consequently, $G$ must be in $\DSC$.
	
	$(\ref{dsc2}) \Rightarrow (\ref{dsc4})$ By Observation \ref{observation:und_dir_SC}.
  
    $(\ref{dsc3}) \Rightarrow (\ref{dsc2})$ By Observation \ref{forb-u-d2}.
  
	$(\ref{dsc3}) \Leftrightarrow (\ref{dsc4})$ Since $\forb(\SC)= \{P_4, \co 2P_3, 2K_2\}$.
\end{proof}

%%%%%%%%%%%%%%%%%%%%%%%%%%%%%%%%%%%%%%%%%%%%%%%%%%%%%%%%%%%
\subsection{Oriented simple co-graphs}\label{sec-oscg}
%%%%%%%%%%%%%%%%%%%%%%%%%%%%%%%%%%%%%%%%%%%%%%%%%%%%%%%%%%%

\begin{definition}[Oriented simple co-graphs] \label{oscg}
The class of {\em oriented simple co-graphs} is recursively defined as follows.
\begin{enumerate}[(i)]
\item Every digraph on a single vertex $(\{v\},\emptyset)$,
denoted by $\bullet$, is an oriented simple co-graph.
\item
If $G$ is an oriented simple co-graph $I$ is an edgeless digraph,
then
\begin{inparaenum}[(a)]
\item
$G\oplus I$,
\item
$G\oslash I$,
\item
$I \oslash G$ are  oriented simple co-graphs.
\end{inparaenum}
\end{enumerate}
The class of oriented simple co-graphs is denoted by $\OSC$.
\end{definition}

\begin{table}[h]
\caption{Complementary graphs of $D_9, D_{10}$ and $D_{11}$}
\label{tbl:D9_D10_D11}
  \center
  \begin{tabular}{c c c }
	\begin{tikzpicture}
  % How to define variables:
  %  \def \variable {wert};
  % Y_4
  % nodes
  \draw (0,0.5) \bnode (n1) {};
  \draw (1,0.5) \bnode (n2) {};
  \draw (0,0) \bnode (n3) {};
  \draw (1,0) \bnode (n4) {};

  \draw [-\ahead] (n1) edge (n2);
  \draw [-\ahead] (n2) edge (n1);
  \draw [-\ahead] (n3) edge (n4);
  \draw [-\ahead] (n4) edge (n3);
\end{tikzpicture} & \begin{tikzpicture}
  % How to define variables:
  %  \def \variable {wert};
  % Y_4
  % nodes
  \draw (0,0.5) \bnode (n1) {};
  \draw (1,0.5) \bnode (n2) {};
  \draw (0,0) \bnode (n3) {};
  \draw (1,0) \bnode (n4) {};

  \draw [-\ahead] (n1) edge (n2);
  \draw [-\ahead] (n3) edge (n4);
  \draw [-\ahead] (n4) edge (n3);
\end{tikzpicture} & \begin{tikzpicture}
  % How to define variables:
  %  \def \variable {wert};
  % Y_4
  % nodes
  \draw (0,0.5) \bnode (n1) {};
  \draw (1,0.5) \bnode (n2) {};
  \draw (0,0) \bnode (n3) {};
  \draw (1,0) \bnode (n4) {};

  \draw [-\ahead] (n1) edge (n2);
  \draw [-\ahead] (n3) edge (n4);
\end{tikzpicture} \\
	$\co D_{9}$ & $\co D_{10}$ & $\co D_{11}$
  \end{tabular}
\end{table}

\begin{theorem}\label{char-osc}
Let $G$ be a digraph. The following properties are equivalent:
\begin{enumerate}
  \item\label{osc1} $G$ is an oriented  simple co-graph.
  \item\label{osc2} $G\in \free(\{D_1,D_5,D_8, Q_7, \co D_{11},\overleftrightarrow{K_2}\})$.
  \item\label{osc3} $G\in \free(\{D_8, Q_7, \co D_{11},\overleftrightarrow{K_2}\})$ and $G$ is transitive.
\end{enumerate}
\end{theorem}

\begin{proof}
	$(\ref{osc1}) \Rightarrow (\ref{osc2})$
	Let $G$ be in $\OSC$. Since $\OSC \subset \DSC$ and $\OSC \subset \OC$ we know that the forbidden induced subdigraphs of $\DSC$,
	as well as the forbidden induced subdigraphs of $\OC$, must also be excluded in $\OSC$. Consequently, none of $D_1,D_5,D_8, Q_7, \co D_{11},\overleftrightarrow{K_2}$ can be induced subdigraphs of $G$.
	Since the class $\OSC$ is closed under taking induced subdigraphs this holds for every digraph of the class.

	$(\ref{osc2}) \Rightarrow (\ref{osc1})$
	Let $G$ contain none of the digraphs from above as induced subdigraph. Since $G$ has no induced $D_1, D_5, D_8$ and $\overleftrightarrow{K_2}$, by Theorem \ref{char-oc} $G$ must be an oriented co-graph and thus, $G$ must have been built by a disjoint union or a series composition.
	Further, we know from the proof of Theorem \ref{char-dsc} which forbidden subdigraphs are exactly leading to the restriction of the disjoint union and the order composition. Since $G$ has no induced $\overleftrightarrow{K_2}$, the list reduces to $ Q_7$ and $\co D_{11}$. With the same argumentation as in the proof of Theorem \ref{char-dsc}, every other digraph that is constructed by a directed union or an order composition is either legit for this class, or contains one of these forbidden digraphs as subdigraph. Consequently, this leads exactly to the definition of $\OSC$, such that $G$ must be in $\OSC$.

	$(\ref{osc2}) \Rightarrow (\ref{osc3})$
	Holds with Lemma \ref{le-tt}.

	$(\ref{osc3}) \Rightarrow (\ref{osc2})$
	The transitivity of $G$ implies that $D_1$ and $D_5$ are excluded, as they do not satisfy the definition of transitivity.
\end{proof}

\begin{observation}
  If $G\in DSC$ then the underlying undirected graph of the symmetric part of $G$ is a simple co-graph and the asymmetric part of $G$ is an oriented simple co-graph.
\end{observation}

This holds since the asymmetric part is exactly build with the same rules like simple co-graphs and the asymmetric part with the rules of \OSC.

%%%%%%%%%%%%%%%%%%%%%%%%%%%%%%%%%%%%%%%%%%%%%%%%%%%%%%%%%%%
\subsection{Directed co-simple co-graphs}\label{sec-dcscg}
%%%%%%%%%%%%%%%%%%%%%%%%%%%%%%%%%%%%%%%%%%%%%%%%%%%%%%%%%%%

We introduce the complementary class of $\DSC$, since it holds that
	$$\begin{array}{lll}
		& \{D_1,\ldots,D_8, Q_1,\ldots,Q_7, \co D_9, \co D_{10}, \co D_{11}\} \\
		\neq & \co \{D_1,\ldots,D_8, Q_1,\ldots,Q_7, \co D_9, \co D_{10}, \co D_{11}\}.
	\end{array}
	$$

\begin{definition}[Directed co-simple co-graphs] \label{dsccg}
The class of {\em directed co-simple co-graphs} is recursively defined as follows.
\begin{enumerate}[(i)]
\item Every digraph on a single vertex $(\{v\},\emptyset)$,
denoted by $\bullet$, is a directed co-simple co-graph.

\item
If $G$ is a  directed co-simple co-graph and $K$ is a  bidirectional complete digraph,
then
\begin{inparaenum}[(a)]
\item
$G\oplus K$,
\item
$G\oslash K$,
\item
$K \oslash G$, and
\item
$G\otimes K$  are   directed co-simple co-graphs.
\end{inparaenum}
\end{enumerate}
The class of directed simple co-graphs is denoted by $\DCSC$.
\end{definition}

\begin{theorem}\label{char-dcsc}
	The following statements are equivalent.
	\begin{enumerate}
		\item $G\in \DCSC$.
		\item $G\in \free(\{D_1,\ldots,D_8, \co Q_1,\ldots,\co Q_7, Q_1, D_9, D_{10}\})$.
        \item $G\in\free(\{D_1,\ldots,D_8, \co Q_1,\co Q_4,\ldots,\co Q_6, Q_1, D_{10}\})$ and $\un(G)$ is a co-simple co-graph.
	\end{enumerate}
\end{theorem}

\begin{proof}
	By Lemma \ref{th-free}, we know that $\co \DSC=\free(\co \{D_1,\ldots,D_8, Q_1,\ldots,Q_7, \co D_9, \co D_{10}, \co D_{11}\})$.
	Since $\DCSC=\co \DSC$, we know by Theorem \ref{char-dsc} that the forbidden induced subdigraphs are $D_1,\ldots, D_8$ since, they are self complementary, 
	and $\co Q_1,\ldots,\co Q_7, Q_1, D_9, D_{10}$.\footnote{Note that $Q_1=D_{11}$.} Consequently, this equivalence holds.
\end{proof}

%%%%%%%%%%%%%%%%%%%%%%%%%%%%%%%%%%%%%%%%%%%%%%%%%%%%%%%%%%%
\subsection{Oriented co-simple co-graphs}\label{sec-ocscg}
%%%%%%%%%%%%%%%%%%%%%%%%%%%%%%%%%%%%%%%%%%%%%%%%%%%%%%%%%%%

\begin{definition}[Oriented co-simple co-graphs] \label{dsocg}
The class of {\em oriented co-simple co-graphs} is recursively defined as follows.
\begin{enumerate}[(i)]
\item Every digraph on a single vertex $(\{v\},\emptyset)$,
denoted by $\bullet$, is an oriented co-simple co-graph.

\item
If $G$ is an orientated co-simple co-graph and $T$ is a transitive tournament,
then
\begin{inparaenum}[(a)]
\item
$G\oplus T$,
\item
$G\oslash T$,  and
\item
$T \oslash G$  are  oriented co-simple co-graphs.
\end{inparaenum}
\end{enumerate}
The class of orientated simple co-graphs is denoted by $\OCSC$.
\end{definition}

It is obvious, that the classes $\OCSC$ and $\OCWQT$ are equal.

%%%%%%%%%%%%%%%%%%%%%%%%%%%%%%%%%%%%%%%%%%%%%%%%%%%%%%%%%%%
\subsection{Directed threshold graphs}\label{sec-dtg}
%%%%%%%%%%%%%%%%%%%%%%%%%%%%%%%%%%%%%%%%%%%%%%%%%%%%%%%%%%%

The class of threshold graphs was introduced by Chvátal and Hammer, see \cite{CH73, CH77}.
Some possible definitions of directed threshold graphs can be found in \cite{GR19c}, 
they work as follows.

\begin{definition}[Directed threshold graphs]\label{dthg}
The class of {\em directed threshold graphs} is recursively defined as follows.
\begin{enumerate}[(i)]
\item Every digraph on a single vertex $(\{v\},\emptyset)$,
denoted by $\bullet$, is a directed threshold graph.

\item If $G$ is a directed threshold graph,
then
\begin{inparaenum}[(a)]
\item
$G\oplus \bullet$,
\item
$G\oslash \bullet$,
\item
$\bullet\oslash G$, and
\item $G\otimes \bullet$ are directed threshold graphs.
\end{inparaenum}
\end{enumerate}
The class of directed threshold graphs is denoted by $\DT$.
\end{definition}

The recursive definition of directed and undirected threshold  graphs lead to the following
observation.

\begin{observation}\label{obs-t1}
For every directed threshold graph $G$ the
underlying undirected graph $\un(G)$ is a threshold graph.
\end{observation}

There are also several other ways to define directed threshold graphs.
Not only by forbidden induced subdigraphs, but also by graph parameters
as directed linear NLC-width and directed neighbourhood-width,
see \cite{GR19c} for their definition.

\begin{theorem}\label{t-dtp}
Let $G$ be a digraph. The following properties are equivalent:
\begin{enumerate}
	\item \label{dteq1}$G$ is a directed threshold graph.
	\item \label{dteq2}$G\in \free(\{D_1, \dots, D_{15}, \co D_{11},
		  \co D_{10},
		  \co D_9\})$.

	\item\label{dteq2aa} $G\in\free(\{D_1,D_2,D_3,D_4,D_5,D_6,D_{10},D_{11},D_{13},D_{14},D_{15}\})$
 and $\un(G)\in\free(\{P_4,2K_2,C_4\})$.
	\item \label{dteq2a} $G\in\free(\{D_1,D_2,D_3,D_4,D_5,D_6,D_{10},D_{11},D_{13},D_{14},D_{15}\})$
 and $\un(G)$ is a threshold graph.

	\item \label{dteq3}$G$ has directed linear NLC-width 1.
	\item \label{dteq4}$G$ has directed neighbourhood-width 1.
\item \label{dteq4a} $G$ has directed linear clique-width at most $2$ and $G\in \free(\{D_2,D_3,D_9,D_{10},
D_{12},D_{13},D_{14}\})$.
\item\label{dteq5} $G$ and $\co G$ are both directed trivially perfect graphs.

\end{enumerate}
\end{theorem}

\begin{proof}
(\ref{dteq1}) $\Rightarrow$ (\ref{dteq2}) The given forbidden graphs are not directed threshold graphs
and the set of directed threshold graphs is closed under taking induced subdigraphs.

(\ref{dteq2}) $\Rightarrow$ (\ref{dteq1}) If digraph $G$ does not contain
$D_1,D_2,D_3,D_4,D_5,D_6,D_7,D_8$ (see Table \ref{F-co}) digraph $G$ is a directed co-graph  by \cite{CP06}
and thus has a construction using disjoint union, series composition, and
order composition.
By excluding $D_9$, $D_{10}$, and $D_{11}$ we know that
for every series composition of $G_1$ and $G_2$ either $G_1$
or $G_2$ is bidirectional complete. Thus this subdigraph can also be added by
a number of series operations with one vertex.

Further by excluding
$D_{12}$, $D_{13}$, $D_{14}$, and $D_{15}$  we know that
for every order composition of $G_1$ and $G_2$ either $G_1$
or $G_2$ is a tournament and since we exclude a directed
cycle of length 3 by $D_5$, by Lemma \ref{le-tt} we know that $G_1$
or $G_2$ is a transitive tournament. Thus this subdigraph can also be added by
a number of order operations with one vertex.

By excluding $\co D_{11},\co D_{10},\co D_9$
for every disjoint union of $G_1$ and $G_2$ either $G_1$
or $G_2$ has no edge. Thus this subdigraph can also be added by
a number of disjoint union operations with one vertex.

(\ref{dteq2aa}) $\Leftrightarrow$ (\ref{dteq2a}) Since $\forb(\T)=\{C_4,P_4,2K_2\}$

(\ref{dteq1}) $\Leftrightarrow$ (\ref{dteq3}) and $(\ref{dteq3})\Leftrightarrow(\ref{dteq4})$ \cite{GR19c}. 

(\ref{dteq1}) $\Rightarrow$ (\ref{dteq4a}) and (\ref{dteq4a}) $\Rightarrow$ (\ref{dteq2}) \cite{GR19c}

(\ref{dteq2}) $\Rightarrow$ (\ref{dteq2a}) By Observation \ref{obs-t1}. 

(\ref{dteq2aa}) $\Rightarrow$ (\ref{dteq2}) By Observation \ref{forb-u-d2}

(\ref{dteq1}) $\Rightarrow$ (\ref{dteq5})  $G$ and $\co G$ are both directed threshold graphs and
thus both are  directed trivially perfect graphs.

(\ref{dteq5}) $\Rightarrow$ (\ref{dteq1})
Let $X=\{D_1,\ldots,D_{15}\}$.
By Theorem \ref{ch-dtp}
we know that $G\in \free(X)$ and $\co G \in \free(X)$. Lemma \ref{th-free} implies that
$G\in \free(X)\cap \free(\co X)$.
And again by Lemma \ref{th-free} we obtain
$G\in \free(X\cup \co X)$.
By Observation \ref{obs18} and part (2) of this theorem it holds
that $G$ is a directed threshold graph.
\end{proof}

For directed threshold graphs
Observation \ref{obs18} leads to the next result.

\begin{propostion}
$\DT=\co\DT$
\end{propostion}

Similar as undirected threshold graphs (cf.~\cite{HSS06}), directed threshold graphs can also
be characterized by the existence of a special sequence.

A {\em directed creation sequence} for $G=(V,E)$ with $V=\{v_1,\ldots,v_n\}$ is
a quaternary string $t=t_1,\ldots,t_n$ of length $n$  such that
there is a bijection $v:\{1,\ldots,n\}\to V$ with
\begin{itemize}
\item
$t_i=3$ if  $v(i)$ is a
bi-dominating vertex for the graph induced by $\{v(1),\ldots,v(i)\}$

$t_i=2$ if  $v(i)$ is a
in-dominating vertex for the graph induced by $\{v(1),\ldots,v(i)\}$
\item
$t_i=1$ if  $v(i)$ is a
out-dominating vertex for the graph induced by $\{v(1),\ldots,v(i)\}$
and
\item
$t_i=0$ if  $v(i)$ is an isolated vertex for the graph induced by $\{v(1),\ldots,v(i)\}$.
\end{itemize}
W.l.o.g. we define a single vertex to be a dominating vertex, i.e.~$t_1=1$.
By adapting the result for undirected threshold graphs from
\cite[Fig.\ 1.4]{MP95a} a directed creation sequence
can be found in time  $\bigo(n+m)$. Furthermore this implies that
directed threshold graphs can be recognized in time $\bigo(n+m)$.

%%%%%%%%%%%%%%%%%%%%%%%%%%%%%%%%%%%%%%%%%%%%%%%%%%%%%%%%%%%
\subsection{Oriented threshold graphs}\label{sec-or-tresh}
%%%%%%%%%%%%%%%%%%%%%%%%%%%%%%%%%%%%%%%%%%%%%%%%%%%%%%%%%%%

The class of oriented threshold graphs has been introduced in \cite{Boe18} as follows.

\begin{definition}[Oriented threshold graphs \cite{Boe18}]
An oriented graph $G = (V,A)$ is an {\em oriented threshold graph}, if there exists
an injective weight function  $w : V \to \IR$ and a threshold value $t \in \IR$
such that $(x,y)\in A$ if and only if $|w(x)| + |w(y)| \geq  t$ and $w(x) > w(y)$.
\end{definition}

\begin{theorem}[\cite{Boe18}]
Let $G$ be an oriented digraph. The following properties are equivalent:
\begin{enumerate}
	\item $G$ is an oriented threshold graph.
	\item $G\in\free(\{D_1, D_5\})$ and $\un(G)\in \free(\{2K_2,C_4,P_4\})$.
\item $G$ is a transitive orientation of a threshold graph.
\item  $G$ can be
constructed from the one vertex empty graph by successively adding an
isolated vertex, an out-dominating vertex or an in-dominated vertex.
\end{enumerate}
\end{theorem}

Using our notations we obtain the following definition:

\begin{definition}[Oriented threshold graphs] \label{othgwe}
The class of {\em oriented threshold graphs} is recursively defined as follows.
\begin{enumerate}[(i)]
\item Every digraph on a single vertex $(\{v\},\emptyset)$,
denoted by $\bullet$, is an oriented threshold graph.
\item If $G$ is an  oriented threshold graph,
then
\begin{inparaenum}[(a)]
\item
$G\oplus \bullet$,
\item
$G\oslash \bullet$, and
\item
$\bullet\oslash G$ are oriented threshold graphs.
\end{inparaenum}
\end{enumerate}
The class of oriented threshold graphs is denoted by $\OT$.
\end{definition}

The recursive definition of  oriented and undirected threshold graphs lead to the following
observation.

\begin{observation}\label{obs-ot1}
For every oriented threshold graph $G$ the
underlying undirected graph $\un(G)$ is a co-graph.
\end{observation}

This class can also be defined by forbidden induced subdigraphs. As
it was possible for oriented co-graphs and oriented trivially perfect
graphs, we can use the fact that oriented threshold graphs are
exactly the directed threshold graphs not containing an induced
$\overleftrightarrow{K_2}$:

\begin{theorem}\label{char-oc}
Let $G$ be a digraph. The following properties are equivalent:
\begin{enumerate}
	\item\label{ch-ot1} $G$ is an oriented threshold graph.
        \item\label{ch-ot1a} $G$ is an oriented co-trivially perfect graph.
	\item\label{ch-ot2} $G\in\free(\{D_1, D_5, D_8, D_{12}, 2 \overrightarrow{P_2}, \overleftrightarrow{K_2}\})$.

        \item\label{ch-ot3} $G\in\free(\{D_1,D_5,\overleftrightarrow{K_2}\})$
              and $\un(G)\in\free(\{2K_2,C_4,P_4\})$.
        \item\label{ch-ot4} $G\in \free(\{D_1,D_5,\overleftrightarrow{K_2}\})$
              and $\un(G)$ is a threshold graph.

        \item\label{ch-ot6} $G\in \free(\{D_1, D_5,  D_{12}, \overleftrightarrow{K_2}\})$ and $\un(G)\in\free(\{P_4,2K_2\})$.
        \item\label{ch-ot7} $G\in \free(\{D_1, D_5,  D_{12}, \overleftrightarrow{K_2}\})$ and $\un(G)$ is a  co-trivially perfect graph.

        \item\label{ch-ot8} $G\in \free(\{D_1, D_5,  D_{12}, 2 \overrightarrow{P_2}, \overleftrightarrow{K_2}\})$ and $\un(G)\in\free(\{P_4\})$.
        \item\label{ch-ot9} $G\in \free(\{D_1, D_5,  D_{12}, 2 \overrightarrow{P_2}, \overleftrightarrow{K_2}\})$ and $\un(G)$ is a  co-graph.

       \item\label{ch-ot5} $G$ is  transitive and $G\in\free(\{D_8, D_{12},2 \overrightarrow{P_2},\overleftrightarrow{K_2}\})$.
\end{enumerate}
\end{theorem}

\begin{proof}
$(\ref{ch-ot1})\Leftrightarrow(\ref{ch-ot1a})$ By the recursive definition of the classes, which arises
from the restriction of the directed classes to oriented graphs.

$(\ref{ch-ot1})\Rightarrow(\ref{ch-ot2})$ If $G$ is an oriented threshold  graph, then
$G$ is a directed threshold graph and by Theorem \ref{t-dtp} it holds that $G\in \free(\{D_1, \dots, D_{15}, \co D_{11},
		  \co D_{10},
		  \co D_9\})$.
Further $G\in \free(\{\overleftrightarrow{K_2}\})$
because of the missing series composition. This leads to  $G\in\free(\{D_1, D_5, D_8, D_{12}, \overleftrightarrow{K_2}, 2 \overrightarrow{P_2}\})$.

$(\ref{ch-ot2})\Rightarrow(\ref{ch-ot1})$
If $G\in\free(\{D_1, D_5, D_8, D_{12}, \overleftrightarrow{K_2}, 2 \overrightarrow{P_2}\})$, then $G\in \free(\{D_1, \dots, D_{15}, \co D_{11},
		  \co D_{10},
		  \co D_9\})$ and
is a directed threshold graph.
Since $G\in \free(\{\overleftrightarrow{K_2}\})$
there is no series operation in any
construction of $G$ which implies that $G$ is an oriented threshold graph.

(\ref{ch-ot3}) $\Leftrightarrow$ (\ref{ch-ot4}) Since $\forb(\T)=\{C_4,P_4,2K_2\}$.

(\ref{ch-ot6}) $\Leftrightarrow$ (\ref{ch-ot7}) Since $\forb(\CTP)=\{P_4,2K_2\}$.

(\ref{ch-ot8}) $\Leftrightarrow$ (\ref{ch-ot9}) Since $\forb(\C)=\{P_4\}$.

%zwei mal sowas
(\ref{ch-ot2}) $\Rightarrow$ (\ref{ch-ot4}), (\ref{ch-ot2}) $\Rightarrow$ (\ref{ch-ot7}), (\ref{ch-ot2}) $\Rightarrow$ (\ref{ch-ot9}) By Observation \ref{obs-ot1}.

(\ref{ch-ot3}) $\Rightarrow$ (\ref{ch-ot2}), (\ref{ch-ot6}) $\Rightarrow$ (\ref{ch-ot2}), (\ref{ch-ot8}) $\Rightarrow$ (\ref{ch-ot2}) By Observation \ref{forb-u-d2}

$(\ref{ch-ot2})\Rightarrow(\ref{ch-ot5})$ By Lemma \ref{le-co-t}
we know that $G$ is transitive.

$(\ref{ch-ot5})\Rightarrow(\ref{ch-ot2})$
 If $G$ is transitive, then  $G\in\free(\{D_1, D_5\})$.
\end{proof}

The proof for the next Theorem can be done very similar
to the proof of Theorem \ref{th-ori-tp}.

\begin{theorem}
A graph $G$ is a threshold graph if and only if there exists an orientation $G'$ of  $G$
such that $G'$ is an oriented threshold graph.
\end{theorem}

\begin{observation}
  If $G\in DT$ then the underlying undirected graph of the symmetric part of $G$ is a threshold graph and the asymmetric part of $G$ is an oriented threshold graph.
\end{observation}

This holds since the asymmetric part is exactly build with the same rules like threshold graphs and the asymmetric part with the rules of \OT. 

Similar to directed threshold graphs every oriented co-graph can be defined by a sequence with only three operations, which can be used to give a linear time recognition algorithm.

%%%%%%%%%%%%%%%%%%%%%%%%%%%%%%%%%%%%%%%%%%%%%%%%%%%%%%%%%%%
\subsection{Threshold Digraphs and Ferres Digraphs}\label{sec-dir-tresha}
%%%%%%%%%%%%%%%%%%%%%%%%%%%%%%%%%%%%%%%%%%%%%%%%%%%%%%%%%%%

Threshold digraphs are a different way to define a directed version
of threshold graphs. This idea comes from \cite{CDMS14} and uses a
definition of forbidden subdigraphs.
A {\em $2$-switch} is is a vertex set $\{w,x,y,z\}$ such that there
are edges $(w,x)$ and $(y,z)$ but there are no edges $(w,z)$ and
$(y,x)$, see Figure \ref{F-ex-d}. Examples for a 2-switch are
$\co D_{10},\co D_9,
\co D_{11}$ and $\overrightarrow{P_4}$.
On the other hand, there are no 2-switches within directed threshold graphs.

\begin{observation}\label{obs-dt-2-s}
Let $G$ be a directed threshold graph, the $G$ does not contain
a $2$-switch.
\end{observation}

\begin{definition}[Threshold digraphs  \cite{CDMS14}]
A digraph $G$ is a {\em threshold digraph} if it does not contain
a $2$-switch  nor a $D_5$ as induced subdigraph.

The class of threshold digraphs is denoted by $\TD$.
\end{definition}

This class is not very useful for our  directed co-graph hierarchy, as they
are incomparable to most of the graph classes in there, though it is
a superclass of directed threshold graphs (see Section \ref{overview}).

\begin{figure}[tp]
\centering
\parbox[b]{.44\textwidth}{
\caption{A 2-switch. All vertices are distinct. Solid
arcs must appear in the digraph and dashed arcs must not appear in the
digraph. If an arc is not given, then it may or may not be present.
}
\medskip
\centerline{\includegraphics[width=0.15\textwidth]{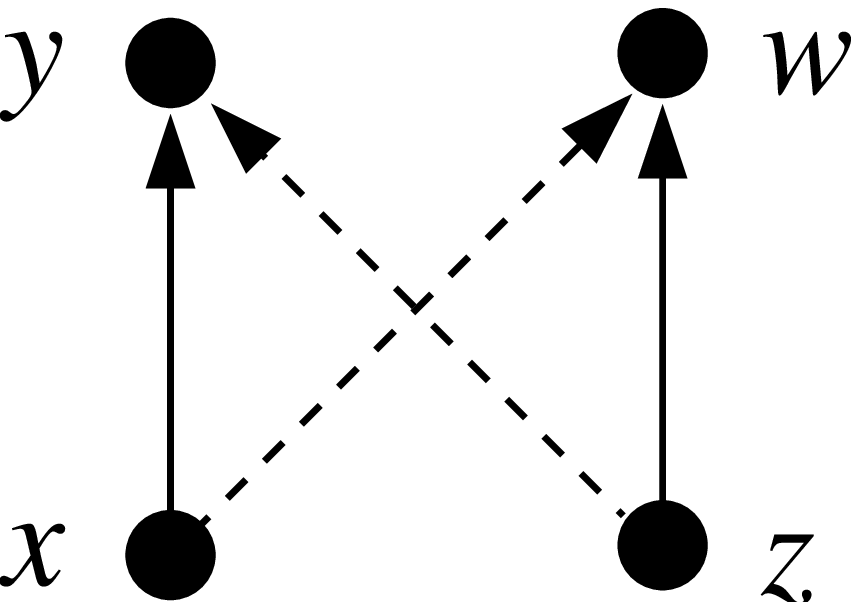}}

\label{F-ex-d}
}
\hfill
\parbox[b]{.44\textwidth}{
\caption{An alternating 4-anticircuit.
The vertices are not necessarily distinct but $x\neq z$ and
$y\neq w$. The solid arrows represent the presence of an arc
and a dashed arrow its absence.
}
\medskip
\centerline{\includegraphics[width=0.15\textwidth]{ferr.eps}}

\label{ferr}
}
\end{figure}

A well studied class of digraphs are Ferres digraphs, see  \cite[Chapter 2]{MP95a} for a
survey.
Ferres digraphs are introduced by Riguet in \cite{Rig51}.
In their first definition, Ferres digraphs where defined
on directed graphs including loops. As our subclasses of
directed co-graphs do not use loops, only Ferres digraphs without loops
will be used here.
An {\em alternating 4-anticircuit} consists of vertices $x,y,z,w$, not necessarily distinct but $x\neq z$ and
$y\neq w$, satisfying
$(x,y),(z,w)\in A$ and $(x,w),(z,y)\not\in A$ (cf. Figure \ref{ferr}).

\begin{definition}[Ferres digraphs]
A digraph is a {\em Ferres digraph} if it does not contain
an alternating 4-anticircuit.

The class of Ferres digraphs is denoted by $\FD$.
\end{definition}

By considering all possible equalities of vertices in an alternating 4-anticircuit
an equivalent characterization is obtained by  
$G\in \free(\{D_1,\overleftrightarrow{K_2}\})$ and $G$ does not contain a $2$-switch,
see  \cite[Figure 2.2]{MP95a} and our restriction to digraphs without loops.
This class is comparable to oriented threshold graphs, but not to any other graph class in
the directed co-graph hierarchy, as we will see in  Section \ref{rel}.

%%%%%%%%%%%%%%%%%%%%%%%%%%%%%%%%%%%%%%%%%%%%%%%%%%%%%%%%%%%
\subsection{Overview} \label{overview}
%%%%%%%%%%%%%%%%%%%%%%%%%%%%%%%%%%%%%%%%%%%%%%%%%%%%%%%%%%%

In Table \ref{t-dc} we summarize directed co-graphs and their subclasses.

Since directed co-graphs and all defined subclasses are
hereditary, by Theorem \ref{th-free-forb} there exist sets of minimal forbidden induced subdigraphs.
The given theorems even show this finite sets of minimal forbidden induced subdigraphs for the different classes.
These characterizations lead polynomial time recognition algorithms for
the corresponding graph classes.

\begin{table}[h!]
\caption{Overview on subclasses of directed co-graphs. By $G_1$ and $G_2$ we denote graphs
of the class $X$, by $I$ we denote an edgeless graph, by $K$ we denote a  bidirectional complete digraph, and
by $T$ we denote a transitive tournament.
Classes and complement classes are listed
between two horizontal lines. That is, only directed co-graphs and directed threshold graphs  are
closed under edge complementation.}
\label{t-dc}
{\footnotesize
\begin{center}
\begin{tabular}{|l|c|ccccc|l|}
\hline
class $X$ & notation &\multicolumn{5}{c|}{operations}   & $\forb(X)$  \\
\hline
\hline
directed co-graphs              & $\DC$       & $\bullet$   & $G_1\oplus G_2$&   $G_1\oslash G_2$  &&$G_1\otimes G_2$ &  $D_1, \dots, D_8$    \\
\hline
oriented co-graphs              & $\OC$        &  $\bullet$   &$G_1\oplus G_2$ &   $G_1\oslash G_2$ & & &  $D_1, D_5, D_8, \overleftrightarrow{K_2}$    \\
\hline
directed trivially perfect      & $\DTP$        & $\bullet$ & $G_1\oplus G_2$ & $G_1 \oslash \bullet$ & $\bullet \oslash G_2$ & $G_1 \otimes \bullet$ & $D_1, \dots ,D_{15}$ \\
\hline
oriented trivially perfect      & $\OTP$       & $\bullet$ & $G_1\oplus G_2$ & $G_1 \oslash \bullet$ & $\bullet \oslash G_2$ & & $D_1,D_5,D_8,\overleftrightarrow{K_2},D_{12}$ \\
\hline
directed co-trivially perfect   &    $\DCTP$     &$\bullet$ &  $G_1 \oplus \bullet$ & $G_1 \oslash \bullet$ & $\bullet \oslash G_2$ & $G_1\otimes G_2$ & $D_1, \dots ,D_{8}$,  $D_{12}, \dots ,D_{15}$\\
                                &                      &           &                   &                    &&  &             $\co D_{11},\co D_{10},\co D_9$\\
\hline
oriented co-trivially   &   $\OCTP$     &$\bullet$   & $G_1 \oplus \bullet$ & $G_1 \oslash \bullet$ & $\bullet \oslash G_2$ &   & $D_1, D_5, D_8, D_{12},$\\
perfect $*$ &&&&&&& $\co D_{11}, \overleftrightarrow{K_2}$ \\
\hline
directed weakly & $\DWQT$      &$I$   & $G_1\oplus G_2$ & $G_1 \oslash I$ & $I \oslash G_2$ & $G_1 \otimes I$& $D_1,\ldots,D_8, Q_1,\ldots,Q_7$ \\
quasi threshold &&&&&&& \\
\hline
oriented weakly & $\OWQT$      &$I$   &   $G_1\oplus G_2$ & $G_1 \oslash I$ & $I \oslash G_2$ & & $D_1,D_5,D_8, \overleftrightarrow{K_2},Q_7$\\
quasi threshold &&&&&&& \\

\hline
directed co-weakly & $\DCWQT$      &$K$   & $G_1 \oplus K$&  $G_1 \oslash K$ & $K \oslash G_2$ &  $G_1\otimes G_2$ & $D_1,\ldots,D_8,\co Q_1,\ldots,\co Q_7$\\
quasi threshold &&&&&&& \\
\hline
oriented co-weakly & $\OCWQT$      &$T$   &  $G_1 \oplus T$ & $G_1 \oslash T$ & $T \oslash G_2$ & & $D_1,D_5,D_8, \overleftrightarrow{K_2},$\\
quasi threshold &&&&&&& $D_{12}, D_{21},D_{22},D_{23} $ \\
\hline

directed simple co-graphs       & $\DSC$      &$\bullet$   & $G_1\oplus I$ & $G_1 \oslash I$ & $I \oslash G_2$ & $G_1 \otimes I$& $D_1,\ldots,D_8,Q_1,\ldots,Q_7,$\\
&&&&&&& $ \co D_9, \co D_{10}, \co D_{11}$ \\
\hline
oriented simple co-graphs       & $\OSC$       &$\bullet$   &  $G_1\oplus I$ & $G_1 \oslash I$ & $I \oslash G_2$ & & $D_1,D_5,D_8, \overleftrightarrow{K_2},Q_7, \co D_{11}$ \\
\hline

directed co-simple co-graphs       &  $\DCSC$       &$\bullet$   & $G_1\oplus K$ & $G_1 \oslash K$ & $K \oslash G_2$ & $G_1 \otimes K$& $D_1,\ldots,D_8,Q_1,$\\
&&&&&&& $\co Q_1,\ldots,\co Q_7, D_{9}, D_{10}$\\
\hline
oriented co-simple co-graphs       & $\OCSC$      &$\bullet$   &  $G_1\oplus T$ & $G_1 \oslash T$ & $T \oslash G_2$ & & $D_1,D_5,D_8, \overleftrightarrow{K_2},$\\
&&&&&&& $D_{12}, D_{21}, D_{22}, D_{23}$ \\
\hline

directed threshold graphs       & $\DT$        & $\bullet$ & $G_1 \oplus \bullet$ & $G_1 \oslash \bullet$ & $\bullet \oslash G_2$ & $G_1 \otimes \bullet$ & $D_1, \dots ,D_{15}$ \\
                                & & & & & & & $\co D_{11},\co D_{10},\co D_9$ \\
\hline
oriented threshold graphs  $*$     & $\OT$   & $\bullet$ & $G_1 \oplus \bullet$ & $G_1 \oslash \bullet$ & $\bullet \oslash G_2$ & & $D_1,D_5,D_8,\overleftrightarrow{K_2},D_{12},\co D_{11}$ \\
\hline
threshold digraphs              & $\TD$ & &&&&&$D_5$, 2-switch\\
\hline
Ferres digraphs                 & $\FD$ & &&&&& $D_1,\overleftrightarrow{K_2}$,  2-switch\\
\hline
\hline
edgeless digraphs                 &               & $\bullet$     &  $G_1\oplus \bullet$  &&&& $\overrightarrow{P_2}$, $\overleftrightarrow{K_2}$  \\
\hline
bidirectional      &               & $\bullet$     &   &&& $G_1\otimes \bullet$ & $\overrightarrow{P_2}$, $\co \overleftrightarrow{K_2}$  \\
complete digraphs  &&&&&&& \\
\hline
\hline
transitive tournaments        & $\TT$ &$\bullet$  &  &   $G_1 \oslash \bullet$& && $2\overleftrightarrow{K_1},\overleftrightarrow{K_2},D_5$\\
\hline
\hline
bidirectional complete     &  &  &  & &  &$I\otimes I$& co-$\overleftrightarrow{P_3},\overrightarrow{P_2},\overleftrightarrow{K_3}$\\
bipartite digraphs   &&&&&&& \\
\hline
disjoint union of  &   &  & $K\oplus K$ & &  &&  $\overleftrightarrow{P_3},\overrightarrow{P_2},\overleftrightarrow{I_3}$\\
2 bidirectional complete  &&&&&&& \\
\hline
\hline
series composition of &  & $I$ &  & &  &$G_1\otimes I$&  co-$\overleftrightarrow{P_3},\overrightarrow{P_2}$\\
stable sets  &&&&&&& \\
\hline
disjoint union of    &   & $K$ & $G_1\oplus K$ & &  &&  $\overleftrightarrow{P_3},\overrightarrow{P_2}$\\
bidirectional complete  &&&&&&& \\
\hline
\end{tabular}
\end{center}
}
\end{table}

%%%%%%%%%%%%%%%%%%%%%%%%%%%%%%%%%%%%%%%%%%%%%%%%%%%%%%%%%%%
\subsection{Relations} \label{rel}
%%%%%%%%%%%%%%%%%%%%%%%%%%%%%%%%%%%%%%%%%%%%%%%%%%%%%%%%%%%

Our forbidden subdigraphs in combination with Theorem \ref{th-free-sub}
allows us to compare the above graph classes to each other and show
the hierarchy of the subclasses of directed co-graphs, see Figure \ref{fig:overview}.

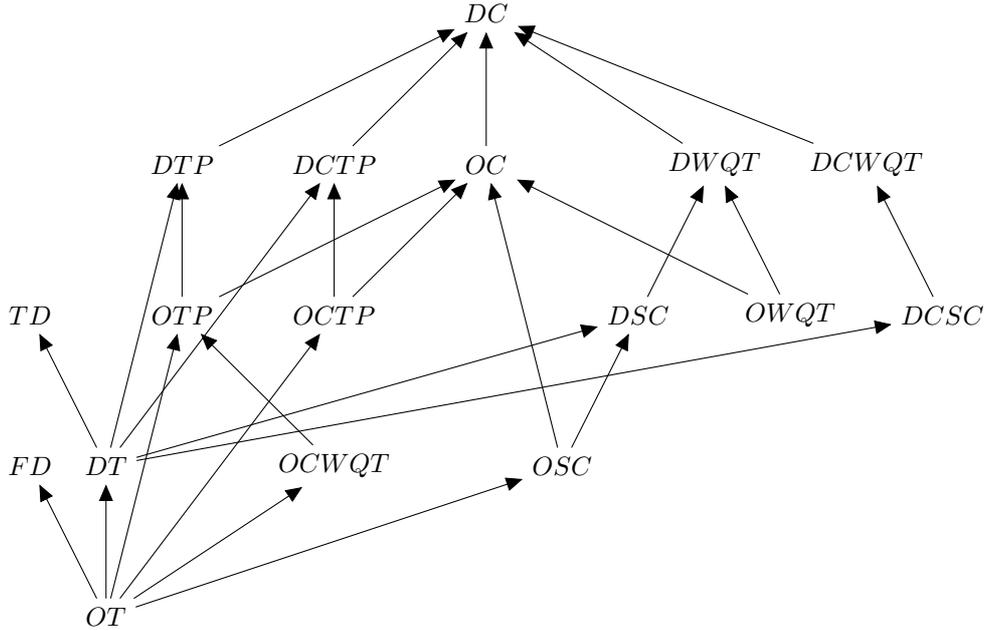
\begin{figure}[h]
	\begin{center}
	\caption{Relations between the subclasses of directed co-graphs. 
	If there is a path from $A$ to $B$, then it holds that $A\subset B$. 
	The classes, that are not connected by a directed path are incomparable.}
	 \medskip
		\begin{tikzpicture}
  % tree before
  % nodes

  % level 0
  \draw (1,0) node (OT){$OT$};
  %level 1
  \draw (4,2) node (OCWQT){$OCWQT$};
  \draw (7,2) node (OSC){$OSC$};
  \draw (1,2) node (DT){$DT$};
  \draw (0,2) node (FD){$FD$};

  % level 2
  \draw (0,4) node (TD){$TD$};
  \draw (2,4) node (OTP){$OTP$};
  \draw (4,4) node (OCTP){$OCTP$};
  \draw (8,4) node (DSC){$DSC$};
  \draw (10,4) node (OWQT){$OWQT$};
  \draw (12,4) node (DCSC){$DCSC$};

  % level 3
  \draw (2,6) node (DTP){$DTP$};
  \draw (4,6) node (DCTP){$DCTP$};
  \draw (6,6) node (OC){$OC$};
  \draw (9,6) node (DWQT){$DWQT$};
  \draw (11,6) node (DCWQT){$DCWQT$};

  % level 4
  \draw (6,8) node (DC){$DC$};

  % Kanten
  \draw [-\ahead] (OT) edge (DT);
  \draw [-\ahead] (OCWQT) edge (OTP);
  \draw [-\ahead] (OT) edge (OTP);
  \draw [-\ahead] (OSC) edge (DSC);
  \draw [-\ahead] (OSC) edge (OC);
  \draw [-\ahead] (DT) edge (DTP);
  \draw [-\ahead] (OTP) edge (DTP);
  \draw [-\ahead] (OCTP) edge (DCTP);
  \draw [-\ahead] (OTP) edge (OC);
  \draw [-\ahead] (OCTP) edge (OC);
  \draw [-\ahead] (DSC) edge (DWQT);
  \draw [-\ahead] (OWQT) edge (DWQT);
  \draw [-\ahead] (OWQT) edge (OC);
  \draw [-\ahead] (DCSC) edge (DCWQT);
  \draw [-\ahead] (DTP) edge (DC);
  \draw [-\ahead] (DCTP) edge (DC);
  \draw [-\ahead] (OC) edge (DC);
  \draw [-\ahead] (DWQT) edge (DC);
  \draw [-\ahead] (DCWQT) edge (DC);

  \draw [-\ahead] (DT) edge (DSC);
  \draw [-\ahead] (DT) edge (DCSC);
  \draw [-\ahead] (OT) edge (OSC);
  \draw [-\ahead] (OT) edge (OCWQT);

  \draw [-\ahead] (DT) edge (DCTP);
  \draw [-\ahead] (OT) edge (OCTP);

  \draw [-\ahead] (OT) edge (FD);
  \draw [-\ahead] (DT) edge (TD);

\end{tikzpicture}
		
		\label{fig:overview}
	\end{center}
\end{figure}

%%%%%%%%%%%%%%%%%%%%%%%%%%%%%%%%%%%%%%%%%%%%%%%%%%%%%%%%%%%%%%%%%%%%%%%%%%%
\section{Conclusions and Outlook}\label{sec-con}
%%%%%%%%%%%%%%%%%%%%%%%%%%%%%%%%%%%%%%%%%%%%%%%%%%%%%%%%%%%%%%%%%%%%%%%%%%%

We introduced several new digraph classes. All these classes
are subsets of directed co-graphs which have been defined by  Bechet et al.\ in \cite{BGR97}
and supersets of oriented threshold graphs defined by Boeckner in \cite{Boe18}.
Further, we consider the ideas of Cloteaux
et al.\ \cite{CDMS14} and Ferres digraphs \cite{Rig51}.

The given characterizations by forbidden induced subdigraphs
lead to polynomial time  recognition algorithms for
the corresponding graph classes. In Section \ref{sec-dtg} and Section  \ref{sec-or-tresh}
we suggested linear time methods for the recognition
of directed and oriented threshold graphs.
For the other classes it remains to find more efficient algorithms
for this purpose.

For directed co-graphs we have shown in \cite{GR18c}
that the directed path-width equals to the directed tree-width and how
to compute this value in linear time. Moreover,
the digraph width parameters directed feedback
vertex set number, cycle rank, DAG-depth, and DAG-width
can be computed in linear time for  directed co-graphs \cite{GKR19f}.
It remains to verify the relation of these parameters
restricted to threshold digraphs and Ferres digraphs.

In \cite{JRST01} the directed union was introduced as a generalization
of the disjoint union and the order composition. In \cite{GR18b} we
considered digraphs which can be defined by
disjoint union, order composition, directed union, and series composition
of two directed graphs. The set of all these digraphs is denoted
by {\em extended directed co-graphs}, since it generalizes
directed co-graphs. In \cite{GR18b} we showed that the
result of \cite{GR18c}  even holds for extended directed co-graphs.
For the class of extended directed co-graphs it remains to show how to compute an
ex-di-co-tree and to find forbidden subdigraph characterizations. Furthermore,
it also has a number of interesting subclasses beside those given in this
paper.

%%%%%%%%%%%%%%%%%%%%%%%%%%%%%%%%%%%%%%%%%%%%%%%%%%%%%%%%%%%%%%%%%%%%%%
\section{Acknowledgements} \label{sec-a}
%%%%%%%%%%%%%%%%%%%%%%%%%%%%%%%%%%%%%%%%%%%%%%%%%%%%%%%%%%%%%%%%%%%%%%
This work was funded in part by the Deutsche Forschungsgemeinschaft
(DFG, German Research Foundation) -- 388221852

%%%%%%%%%%%%%%%%%%%%%%%%%%%%%%%%%%%%%%%%%%%%%%%%%%%%%%%%%%%%%%%%%%%%%%%%%%%
%%%%%%%%%%%%%%%%%%%%%%%%%%%%%%%%%%%%%%%%%%%%%%%%%%%%%%%%%%%%%%%%%%%%%%%%%%%
%%%%%%%%%%%%%%%%%%%%%%%%%%%%%%%%%%%%%%%%%%%%%%%%%%%%%%%%%%%%%%%%%%%%%%%%%%%

%\bibliographystyle{alpha}
%\bibliography{/home/gurski/bib.bib}

\newcommand{\etalchar}[1]{$^{#1}$}

\end{document}